\documentclass[conference]{IEEEtran}
\IEEEoverridecommandlockouts

\usepackage{wrapfig}
\usepackage{ragged2e}
\usepackage[table]{xcolor}
\usepackage{cite}
\usepackage{siunitx}
\usepackage{amsmath,amssymb,amsfonts}
\usepackage{algorithm}
\usepackage{balance}
\usepackage{graphicx}
\usepackage{textcomp}
\usepackage{xcolor}
\usepackage{colortbl}
\usepackage{enumitem}
\usepackage{adjustbox}
\usepackage{multirow}
\usepackage{threeparttable}
\usepackage{subfigure}
\usepackage{mathtools}
\usepackage{makecell}
\usepackage{tcolorbox}
\usepackage[noend]{algpseudocode}
\usepackage{nomencl} 
\usepackage{amsthm}
\makenomenclature     

\newtheorem{theorem}{Theorem}

\newtheorem{lemma}{Lemma}

\usepackage{tikz}

\newcommand{\revise}[1]{\textcolor{black}{#1}}

\def\BibTeX{{\rm B\kern-.05em{\sc i\kern-.025em b}\kern-.08em
    T\kern-.1667em\lower.7ex\hbox{E}\kern-.125emX}}

\begin{document}

\title{
Time-Varying Vector Field Compression with Preserved Critical Point Trajectories
}

\author{\IEEEauthorblockN{Mingze Xia}
\IEEEauthorblockA{
\textit{Oregon State University}\\
Corvallis, OR, USA \\
xiami@oregonstate.edu} \\
\IEEEauthorblockN{Bei Wang}
\IEEEauthorblockA{
\textit{University of Utah}\\
Salt Lake City, UT, USA \\
beiwang@sci.utah.edu}
\and
\IEEEauthorblockN{Yuxiao Li}
\IEEEauthorblockA{
\textit{The Ohio State University}\\
Columbus, OH, USA \\
li.14025@osu.edu} \\
\IEEEauthorblockN{Xin Liang\IEEEauthorrefmark{1}\thanks{*Corresponding author: Xin Liang, School of Electrical Engineering and Computer Science, Oregon State University, Corvallis, OR 97331.}}
\IEEEauthorblockA{
\textit{Oregon State University}\\
Corvallis, OR, USA \\
lianxin@oregonstate.edu}
\and
\IEEEauthorblockN{Pu Jiao}
\IEEEauthorblockA{
\textit{University of Kentucky}\\
Lexington, KY, USA \\
pujiao@uky.edu} \\
\IEEEauthorblockN{Hanqi Guo}
\IEEEauthorblockA{\textit{The Ohio State University} \\
Columbus, OH, USA \\
guo.2154@osu.edu}
}

\maketitle

\begin{abstract}
Scientific simulations and observations are generating massive volumes of time-varying vector field data, posing significant challenges for long-term storage and data transmission. Lossy compression is widely regarded as a promising approach for reducing data volume, as lossless methods typically achieve only modest compression ratios and therefore provide limited reduction. However, directly applying existing lossy compression techniques to time-varying vector fields can introduce undesirable distortions in critical-point trajectories, which encode essential structural properties of the underlying field. 
In this work, we present an efficient lossy compression framework that exactly preserves all critical-point trajectories in time-varying vector fields. Our contributions are threefold. First, we extend the theory of critical-point preservation from the spatial domain to space-time and develop a corresponding compression framework to guarantee trajectory preservation. Second, we introduce a semi-Lagrangian predictor to more effectively exploit spatiotemporal correlations in advection-dominated regions, and integrate it with the classical Lorenzo predictor to further improve compression efficiency. Third, we evaluate the proposed approach against state-of-the-art lossy and lossless compressors on four real-world scientific datasets. Experimental results show that our method achieves compression ratios of up to $124.48\times$ while effectively preserving all critical-point trajectories. This compression ratio is up to $56.07\times$ higher than the best-performing lossless compressors. In contrast, none of the existing lossy compressors can preserve all critical-point trajectories at comparable compression ratios. 
\end{abstract}


\section{Introduction}
\label{sec:introduction}

Large-scale scientific simulations and observations produce massive amounts of data that overwhelm data storage and transmission systems. 
For instance, direct numerical simulation (DNS) of forced isotropic turbulence on the Frontier supercomputer~\cite{frontier} generates 0.5 PB of data in a single run~\cite{yeung2025small}, and the ERA5 Atmospheric Reanalysis, which records climate data from 1979 to the present, has already accumulated to 5 PB~\cite{hersbach2020era5}. 
Such large data sizes not only make it difficult for producers to store the data but also hinder researchers from accessing it remotely.

Time-varying vector fields, such as flow velocity in DNS and atmospheric wind in ERA5, constitute a significant component of scientific data.
These data describe how vectors change over both space and time, and are widely used to compute derived quantities and extract features of interest~\cite{gunther2017generic, friederici2015finite,fritschi2019visualizing,rimensberger2019visualization}. 
As such, enabling timely access to time-varying vector fields is a pressing need across many scientific and engineering disciplines.
This often requires efficient data compression due to unprecedented data sizes. 

While lossless compression is the most general way to reduce data size, it is inefficient for floating-point scientific data because of the random mantissas.
As demonstrated in prior work~\cite{son2014data, lindstrom2017error}, generic lossless compressors, such as GZIP~\cite{gzip} and ZSTD~\cite{zstd}, achieve modest compression ratios of less than 2$\times$ on most scientific datasets. 
Such limited compression is insufficient to address the growing challenges of scientific data storage and transmission. 
Alternatively, error-controlled lossy compression~\cite{zhao2021optimizing, liang2022sz3, zfp, ainsworth2019multilevel, liang2021mgard+, li2023sperr, liu2024high} has emerged as a promising approach, as it achieves substantially higher compression ratios while enforcing user-specified error constraints.

While error-controlled lossy compressors provide error guarantees over the raw data, most of them are agnostic to topological features.
Topological features usually contain key information for novel scientific discoveries, and distortions in them may lead to wrong interpretations of the data and, thus, misleading discoveries.
Critical point trajectories are important topological features of time-varying vector fields. They are time-parameterized curves for which the vector value equals zero at every point, and they are widely used for tracking cyclones in climatology~\cite{hewson2010objective}, identifying eddies in oceanology~\cite{nencioli2010vector}, and tracing vortices in computational fluid dynamics~\cite{BERSON200961,CoresofSwirling}.
Given their role as primary analytical primitives, distortions in critical-point trajectories can lead to systematic errors in downstream analytics.
For example, corrupted trajectories may result in incorrect detection of event initiation or termination, inaccurate estimation of feature lifetimes and movement patterns, and inconsistent query results when tracking features across time~\cite{TRICOCHE2002249,1432684,fff}.
Such errors can further propagate into misleading physical interpretations and flawed predictive models~\cite{neu2013imilast, woollings2018blocking, bai2020time,haller2005objective}.
While prior works~\cite{liang2020toward, liang2022toward, xia2024preserving} enable the preservation of critical points during error-controlled lossy compression of single-snapshot data, they may still distort critical-point trajectories, as the spatiotemporal connectivity of critical points is not explicitly considered.

In this work, we present the first highly effective lossy compression framework that exactly preserves critical-point trajectories in time-varying vector fields. Specifically, we extend the theory in~\cite{xia2024preserving} from preserving critical points in the spatial domain to preserving critical-point trajectories in space-time, and develop a lossy compression framework that guarantees exact trajectory preservation. 
To further improve compression efficiency, we propose a semi-Lagrangian predictor that better exploits spatiotemporal correlation in advection-dominated regions, and combine it with the traditional Lorenzo predictor~\cite{ibarria2003out} for enhanced data decorrelation. 
To this end, we evaluate our framework on five real-world scientific datasets with six baselines. 

In summary, our contributions are as follows.

\begin{itemize}
    \item We design and develop the first lossy compression framework that preserves critical-point trajectories in time-varying vector fields, by extending the theory in~\cite{xia2024preserving} to account for critical-point connections in space-time.  
    \item We propose a semi-Lagrangian predictor that leverages particle advection theory to exploit spatiotemporal correlation in time-varying vector fields. We further leverage a mixture of predictors that combines the semi-Lagrangian predictor and the traditional Lorenzo predictor for improved compression efficiency. 
    \item We evaluate the proposed method and compare it with three widely used lossless compressors and three leading lossy compressors using five scientific datasets. 
    Experiments demonstrate that our method faithfully preserves all critical-point trajectories while delivering up to $56\times$ higher compression ratios than lossless compressors. In contrast, existing lossy compressors produce noticeable distortions in critical-point trajectories at similar or lower compression ratios. 
\end{itemize}

The remainder of this paper is organized as follows. Section~\ref{sec:related} reviews related work. Section~\ref{sec:background} provides background on time-varying vector fields, critical-point trajectories, and existing critical-point-preserving compression algorithms. Section~\ref{sec:formulation} formulates the research problem. Section~\ref{sec:method} presents the design of our compression framework for preserving critical-point trajectories. Section~\ref{sec:mop} introduces the semi-Lagrangian predictor and describes its integration with the Lorenzo predictor. Section~\ref{sec:evaluation} reports and analyzes the experimental results. Finally, Section~\ref{sec:conclusion} concludes the paper and outlines directions for future work.

\section{Related Works}
\label{sec:related}

In this section, we review the literature on compression for scientific data and topology-aware compression.

\subsection{Lossy Compression for Scientific Data}
Lossless compression has been widely adopted in various domains such as database systems~\cite{hu2025icde,liu2025bittuner,lixi2024icde,tan2024dcc,tan2024ts}. Traditional lossless compression techniques, such as LZ-family compressors including gzip~\cite{gzip}, Zstd~\cite{zstd}, and LZ4~\cite{lz4}, floating-point-oriented tools like FPZIP~\cite{fpzip}, typically achieve a limited compression ratio on scientific data\cite{lindstrom2017error}. This limitation has motivated the community to adopt error-bounded lossy compression to achieve a higher compression ratio while maintaining user-specified error metrics. Error-bounded lossy compressors usually follow three  steps: decorrelation, quantization, and encoding. Depending on the decorrelation method, we can categorize these compressors into prediction-based compressors and transform-based compressors.  The SZ compressor and its derived variants~\cite{sz16, sz17, sz18, zhao2021optimizing,liu2025hpdc,wuxuan2025ipdps} represent a class of prediction-based compressors. They use predictors such as Lorenzo and regression for decorrelation and the decorrelated data are further quantized and encoded. 
Transform-based compressors leverage specific transforms for data decorrelation, and then perform quantization and encoding in the transformed domains. 
ZFP~\cite{lindstrom2014fixed} is a typical transform-based compressor that divides data into non-overlapped data blocks for independent compression. Other notable transform-based compressors include TTHRESH~\cite{ballester2019tthresh}, SPERR~\cite{li2023lossy}, and FAZ~\cite{liu2023faz}.  
MGARD~\cite{ainsworth2018multilevel, ainsworth2019multilevel, ainsworth2019qoi} is another lossy compressor that occupies a middle ground between prediction-based and transform-based approaches. It performs data decorrelation via multilevel multilinear interpolation with an $L^2$ projection, followed by level-wise quantization. 

Beyond engineering practice, the community has also studied error propagation and numerical stability, including error analyses and statistical modeling of compression error distribution~\cite{lindstrom2017error}. Together, these results support replacing lossless compression with controlled-lossy approaches when application-level quantities of interest (QoIs)~\cite{pu2022qoi,liu2024qoi,wuxuan2025sc} remain within tolerance. In short, given the limited compression ratios of lossless methods on floating-point data, error-bounded lossy compression has become the mainstream solution for scientific data reduction and lays the groundwork for higher-level, structure-aware guarantees. This motivates topology-aware compressors that explicitly preserve structural invariants.

\subsection{Topology-Aware Compression}
Although error-bounded lossy compressors can enforce constraints such as $L^\infty$ or $L^2$ bounds, many scientific analyses are primarily concerned with the topological structures of scalar or vector fields~\cite{guoxi2025gale,guoxi2023atask}. Examples include the number and locations of critical points, the connectivity of level sets (i.e., contours), and the stability of contour trees, merge trees, and Morse-Smale complexes. Purely amplitude-based error bounds can create or annihilate critical points or alter connectivity, undermining feature tracking and subsequent analysis.

For the scalar field data, a compressor building on persistent homology has been proposed~\cite{soler2018topologically}. This approach guarantees that high-persistence features are preserved exactly under a user-specified persistence threshold, while low-persistence noise is simplified and encoded. Methods such as TopoSZ~\cite{toposz} incorporate topological invariants directly into the prediction-quantization pipeline of SZ~\cite{sz-framework}, preserving the contour tree of scalar fields while meeting pointwise error bounds. An generalization of TopoSZ~\cite{nathan2025pvis} abstracts contour-tree preservation into a standalone layer that can be integrated with arbitrary lossy compressors. MSz and its derived variants~\cite{li2024msz,multitierMSZ} are capable of preserving the Morse-Smale complex structure while simultaneously leveraging parallelization to significantly reduce compression time. A complementary line of work~\cite{GeneralizedTopo} performs local or global topological simplification with $L^\infty$ control. Although this method is not compression algorithms per se, it can be integrated with existing compressor frameworks to preserve the topological structure of the data.

The need for topology-aware constraints naturally extends from scalar fields to vector and tensor fields. For example, the cpSZ framework~\cite{liang2022toward} is capable of preserving critical points during lossy compression, while cpSZ-SoS~\cite{xia2024preserving} further extends this framework into a parallelized compressor suitable for large-scale datasets. At the same time, it replaces the original numerical scheme with the Simulation of Simplicity (SoS) method for critical-point computation, thereby avoiding numerical instabilities that can occur in extreme cases. Building upon these foundations, subsequent work~\cite{xia2025tspsz} further extends such preservation to separatrix structures. More recently, TFZ~\cite{TFZ} introduces a topology-aware lossy compression strategy for 2D second-order tensor fields, illustrating that topology-aware compression is evolving from scalar fields toward higher-dimensional and more complex data types.

\section{background and Preliminaries}
\label{sec:background}

In this section, we present the necessary background and preliminaries for our work, including the notation, fundamental concepts, and techniques that underpin the proposed method.

\subsection{Notation and Assumptions}
We assume the target vector field is piecewise linear, which is a common practice in the community~\cite{nielson1991asymptotic, de2008computational, lorensen1987marching}. Following this assumption, the original domain is decomposed into  simplices, each of which can be represented by a linear function. 
Formally, let $D\subset\mathbb{R}^2$ be a polygonal domain endowed with a time-invariant simplicial mesh $\mathcal{M}=(\mathcal{V},\mathcal{E},\mathcal{F})$, where $\mathcal{V}$, $\mathcal{E}$, and $\mathcal{F}$ denote the sets of vertices, edges, and faces, respectively. The vertex positions are fixed in space across discrete times $\mathcal{T}=\{t_0,\dots,t_T\}$. A time-varying 2D vector field is given by
\[
\mathbf{V}(x,y,t)=(u(x,y,t),\,v(x,y,t))\in\mathbb{R}^2.
\]
We then use $\mathcal{V}_k^i=(u_k^i,v_k^i)$ to denote the vector value at vertex $i\in\mathcal{V}$ and time $t_k$.

\subsection{Critical Point}
A critical point is defined as a location where the vector field vanishes. Under the piecewise-linear assumption over a 2D simplicial mesh, a critical point can be computed by solving a linear system expressed in barycentric coordinates:
\begin{equation}
\label{eq:def}
    \begin{bmatrix}
    u_0 & u_1 & u_2\\
    v_0 & v_1 & v_2
    \end{bmatrix}
    \begin{bmatrix}
    \mu_0\\
    \mu_1\\
    \mu_2
    \end{bmatrix} = \mathbf{0} \ \text{and} \ \mu_0 + \mu_1 + \mu_2 = 1,
\end{equation}
where $(\mu_0, \mu_1, \mu_2)$ are the barycentric coordinates, and $(u_i, v_i)$ denote the vector components at the vertices that constitute the cell. A critical point exists within the cell if and only if $0 \leq \mu_k \leq 1$ holds for all $k \in \{0, 1, 2\}$. Figure~\ref{fig:cpexample} illustrates examples of cells with and without a critical point, respectively. 

\begin{figure}[htb]
\centering
\vspace{-1mm}
\includegraphics[width=1.0\columnwidth]{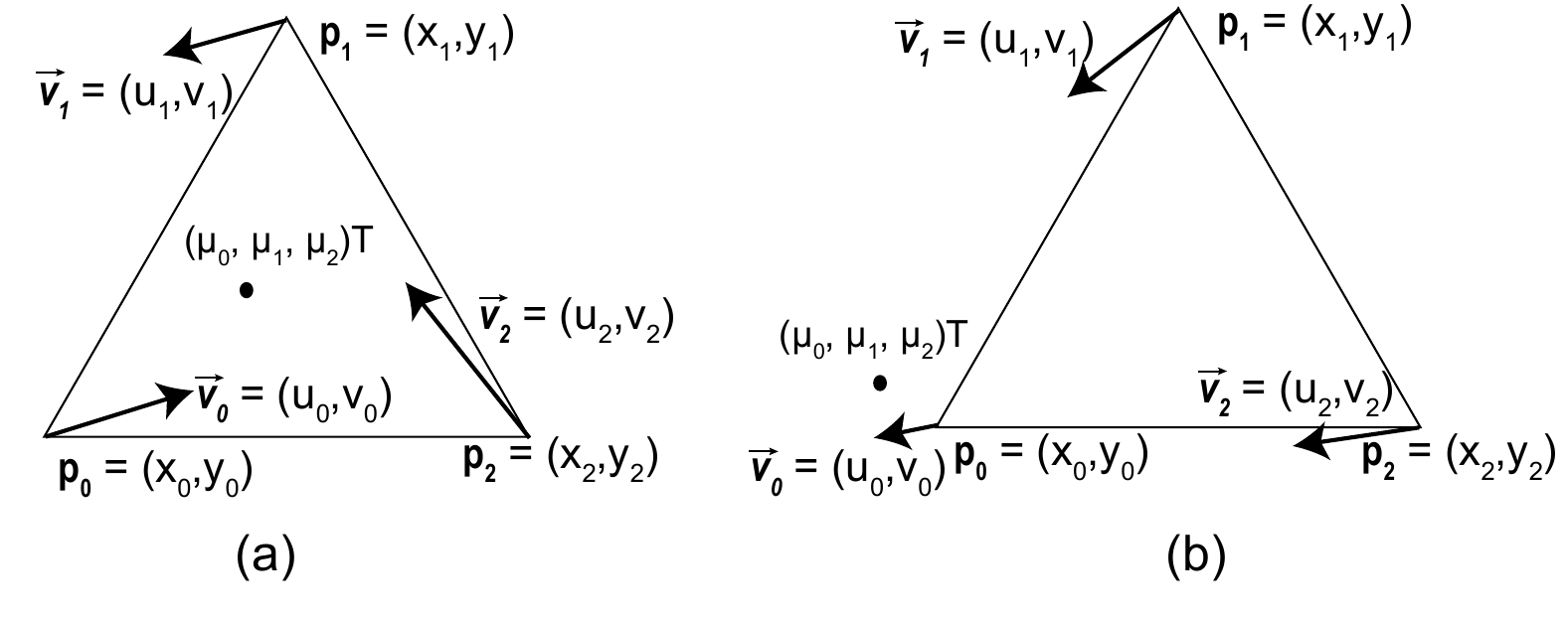}
\vspace{-9mm}
\caption{(a) Barycentric coordinates inside the cell, where a critical point exists; (b) barycentric coordinates outside the cell, where no critical point exists.}
\label{fig:cpexample}
\vspace{-3mm}
\end{figure}

\subsection{Space–Time Extrusion and Tetrahedralization}
\label{sec:bg:3-split}

Existing approaches~\cite{liang2020toward,xia2024preserving} cannot reliably maintain the correctness of critical point tracking paths, as they lack a unified mechanism that couples spatial and temporal domains. Since critical points evolve over space and time through events such as birth, death, split, and merge~\cite{fff}, any lossy compression performed without proper space-time alignment may produce artificial critical points within time slabs. This results in erroneous event identification and, ultimately, incorrect topological transitions.
Figure~\ref{fig:topochange} depicts the evolution of critical points under various topological events and highlights how potential spurious critical points may distort the underlying topological structure, resulting in incorrect event identification and topological changes.

\begin{figure}[htb]
\centering
\vspace{-3mm}
\includegraphics[width=0.95\columnwidth]{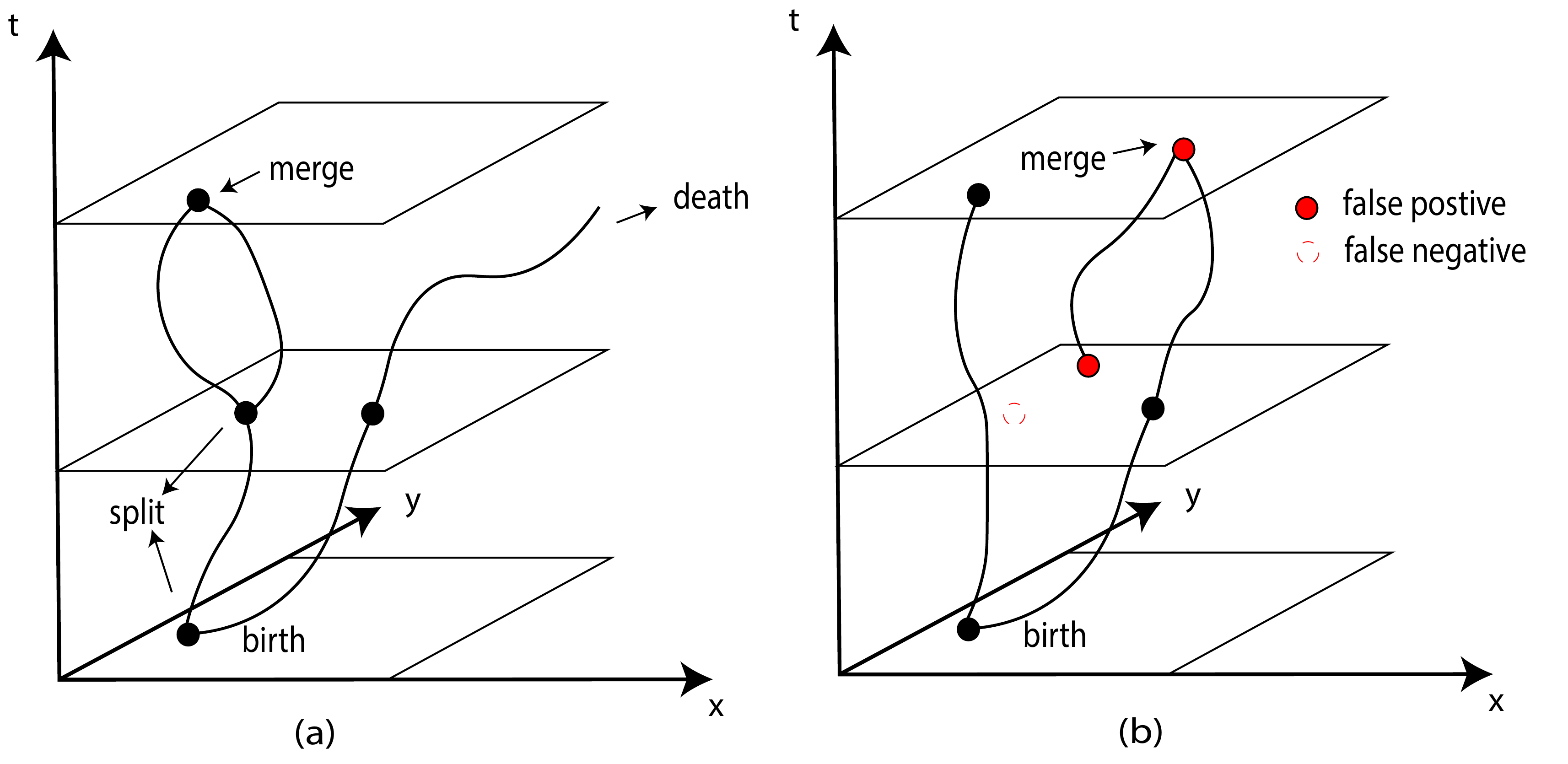}
\vspace{-3mm}
\caption{(a) illustrates the temporal evolution of a critical point and the associated events; (b) demonstrates how lossy compression alters the critical point, resulting in a deviation of its trajectory.}
\label{fig:topochange}
\end{figure}

To address this issue, we adopt the method proposed by De Loera et al.~\cite{tetrahedron_split}, which provides a unified representation of space and time by extruding the spatial domain along the temporal dimension, followed by a consistent space-time decomposition. At each time step $t_k$, we have a consistent triangular mesh $\Delta\in\mathcal{F}$.  
Then we extrude $\Delta$ over the temporal interval $[t_k, t_{k+1}]$ to form a triangular prism.  
This prism is then subdivided into three tetrahedra, denoted as $\tau_1$, $\tau_2$, and $\tau_3$, using a fixed three-tetrahedron splitting scheme: $\tau_1 = (\mathcal{V}_t^a,\, \mathcal{V}_t^b,\, \mathcal{V}_t^c,\, \mathcal{V}_{t+1}^c)$, $\tau_2 = (\mathcal{V}_t^a,\, \mathcal{V}_t^b,\, \mathcal{V}_{t+1}^b,\mathcal{V}_{t+1}^c)$, and $\tau_3 = (\mathcal{V}_t^a,\, \mathcal{V}_{t+1}^a,\, \mathcal{V}_{t+1}^b,\mathcal{V}_{t+1}^c)$.

Within any tetrahedron $\tau\subset D\times[t_k,t_{k+1}]$, the components $(u,v)$ are affine in $(x,y,t)$; hence the joint zero-set $\{\mathbf{V}=0\}$ is either empty or a straight line segment in $\tau$.
As illustrated in Figure~\ref{fig:tet-split}, this subdivision yields a refined and temporally coherent discretization of the data, ensuring a uniform topological structure across consecutive time steps. This consistency is crucial for the accurate detection and tracking of critical points in time-dependent vector fields.

\begin{figure}[htb]
\centering
\vspace{-2mm}
\includegraphics[width=0.95\columnwidth]{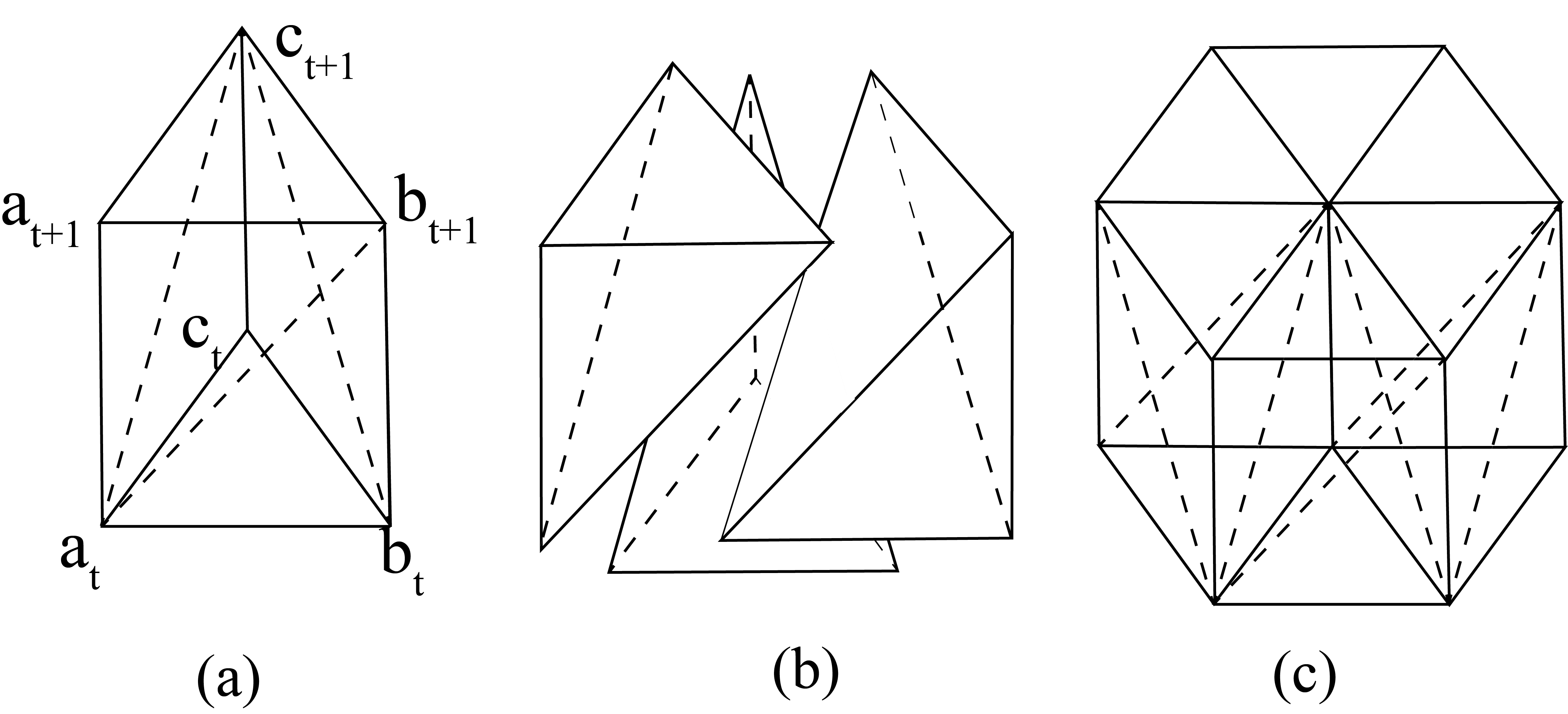}
\vspace{-3mm}
\caption{Space-time extrusion of simplicial mesh. (a) A 2D triangle is extruded along the temporal dimension to form a prism.
(b) A triangular prism is structurally subdivided into three tetrahedra.
(c) A 2D simplicial mesh is extruded into a 3D prismatic mesh and subsequently subdivided.}
\label{fig:tet-split}
\vspace{-4mm}
\end{figure}

\subsection{Critical-Point Trajectory}

A \emph{critical-point trajectory} over $[t_k,t_{k'}]$ is a polyline obtained by concatenating the line segments $\{\mathbf{V}=0\}$ across adjacent tetrahedra in $D\times[t_k,t_{k'}]$.

Under piecewise-linear assumption, the intersection $\{\mathbf{V}=0\}\cap (D\times\{t_k\})$ is governed by a \emph{face-level test}: for a space-time triangular face $f=\{i,j,k\}$, if the origin lies within the convex hull $\mathrm{conv}\{\mathbf{v}_i,\mathbf{v}_j,\mathbf{v}_k\} \subset \mathbb{R}^2$, then the zero set intersects $f$, yielding a slice critical point. 

\subsection{Face-Level Test}
\label{sec:problem:cp-detection}

Consider a triangular face $f$ of the space-time tetrahedral mesh. Let the three incident vertices carry vector values $\{\mathbf{a},\mathbf{b},\mathbf{c}\}\subset\mathbb{R}^2$. Because $\mathbf{V}$ is piecewise-linear, according to \eqref{eq:def}, the restriction of $\mathbf{V}$ to $f$ is the convex hull of $\{\mathbf{a},\mathbf{b},\mathbf{c}\}$ in $\mathbb{R}^2$ , which can be written as:
\[
\{\mathbf{V}|_f\} = \Bigl\{ \alpha\mathbf{a}+\beta\mathbf{b}+\gamma\mathbf{c} \;\big|\; \alpha,\beta,\gamma\ge 0,\; \alpha+\beta+\gamma=1 \Bigr\}.
\]
Define the $2$D scalar cross product $\det(\mathbf{x},\mathbf{y}) = x_x y_y - x_y y_x$.
Let
\[
D_f=\det(\mathbf{a},\mathbf{b})+\det(\mathbf{b},\mathbf{c})+\det(\mathbf{c},\mathbf{a}).
\]
$\operatorname{sign}(x,y,z)=
1$ if $x,y,z>0$,
$-1$ if $x,y,z<0$, and $0$ otherwise.
In general position, the origin $\mathbf{0}$ lies in $\mathrm{conv}\{\mathbf{a},\mathbf{b},\mathbf{c}\}$ iff
\begin{equation}
\label{eq:inside-test}
\operatorname{sign}\!\big\{\det(\mathbf{a},\mathbf{b}),\;
\det(\mathbf{b},\mathbf{c}),\;
\det(\mathbf{c},\mathbf{a})\big\} \neq 0,
\end{equation}
equivalently, the barycentric coefficients of the origin are
\begin{equation}
\label{eq:origin-bary}
(\alpha^\star,\beta^\star,\gamma^\star)=\frac{1}{D_f}\bigl(\det(\mathbf{b},\mathbf{c}),\ \det(\mathbf{c},\mathbf{a}),\ \det(\mathbf{a},\mathbf{b})\bigr),
\end{equation}
which satisfy $\alpha^\star,\beta^\star,\gamma^\star>0$ exactly when \eqref{eq:inside-test} holds.
If \eqref{eq:inside-test} is true, we say \emph{face $f$ is crossed by a critical point}. The point of crossing in the geometric triangular face (in $x\text{--}y\text{--}t$) is the same barycentric combination of the three face vertices.

\subsection{Simulation of Simplicity}

We adopt the \textit{Simulation of Simplicity (SoS)}~\cite{sos} paradigm to avoid ambiguous degeneracies in critical point detection.
Typical degeneracies include: 
(i) $\det(\cdot,\cdot)=0$ in Eq.~\eqref{eq:inside-test}; 
(ii) zero-line intersections exactly passing through a vertex or an edge; and 
(iii) multiple faces (three or more) being simultaneously intersected within one tetrahedron. 
SoS introduces consistent symbolic perturbations (e.g., lexicographic tie-breaking by vertex or time indices) that deterministically select a side, thereby ensuring that 
Lemma~\ref{lem:line} and Theorems~\ref{thm:per-frame}--\ref{thm:track} hold under all configurations. 
Algorithm~\ref{alg:robust_cp} outlines the robust inside-test under SoS for a 2D simplex.

\begin{algorithm}[htb]
\caption{Robust critical point detection in a 2D simplex}
\label{alg:robust_cp}
\renewcommand{\algorithmicrequire}{\textbf{Input:}}
\renewcommand{\algorithmicensure}{\textbf{Output:}}
\renewcommand{\algorithmiccomment}[1]{\hfill\textcolor{gray}{// #1}}
\begin{algorithmic}[1]
\Require $V = \{(v_0, id_0), (v_1, id_1), (v_2, id_2)\}$, where $v_i \in \mathbb{R}^2$ are vertex values and $id_i$ are their global indices 
\State $x \gets (0, 0)$
\State Extract coordinate matrix $P \gets [v_0, v_1, v_2]^\top$ and index array $idx \gets [id_0, id_1, id_2]$
\State $s \gets \textsc{Orient2\_SoS}(P, idx)$
\Comment{evaluate the sign of $3{\times}3$ determinant; flip sign if permutation parity is odd}
\For{$i \gets 0$ to $2$}
  \State $Y \gets P$; \quad $Y[i] \gets x$
  \State $id' \gets idx$
  \State $s_i \gets \textsc{Orient2\_SoS}(Y, id')$
  \If{$s_i \neq s$}
    \State \Return \textbf{false}
  \EndIf
\EndFor
\State \Return \textbf{true}
\end{algorithmic}
\end{algorithm}

\subsection{Derivation of Error Bounds for Critical-Point Preservation}

Previous work~\cite{xia2025tspsz} derived a sufficient but not necessary condition on the absolute error bound to avoid false positive cases (spurious critical points) in critical point preservation. The detailed algorithm is shown in Algorithm~\ref{alg:eb-derive}. We will utilize this derivation in Section~\ref{sec:method} to ensure that the proposed framework can preserve critical-point trajectories.

\begin{algorithm}[htb!]
\caption{Derivation of eb for critical point preservation}\label{alg:eb-derive}
\footnotesize
\renewcommand{\algorithmicrequire}{\textbf{Input:}}
\renewcommand{\algorithmicensure}{\textbf{Output:}}
\renewcommand{\algorithmiccomment}[1]{\hfill\textcolor{gray}{// #1}}
\begin{algorithmic}[1]
\Require $u_0,u_1,u_2,v_0,v_1,v_2$ \Comment{triangle vertices' vector components}
\Ensure $eb$ \Comment{absolute error bound to preserve critical point}
\State $M_0 \gets u_2 v_0 - u_0 v_2$, $M_1 \gets u_1 v_2 - u_2 v_1$
\State $M_2 \gets u_0 v_1 - u_1 v_0$, $M \gets M_0 + M_1 + M_2$
\If{$M = 0$}
  \State \Return $0$
\EndIf
\State $eb \gets \dfrac{|M|}{\,|u_1 - u_0| + |v_0 - v_1|\,}$ 
\If{$|u_1| + |v_1| \neq 0$}
  \State $eb \gets \mathrm{MINF}\!\left( eb,\; \dfrac{|u_1 v_2 - u_2 v_1|}{\,|u_1| + |v_1|\,} \right)$
\Else
  \State \Return $0$
\EndIf
\If{$|u_0| + |v_0| \neq 0$}
  \State $eb \gets \mathrm{MINF}\!\left( eb,\; \dfrac{|u_0 v_2 - u_2 v_0|}{\,|u_0| + |v_0|\,} \right)$
\Else
  \State \Return $0$
\EndIf
\Statex
\If{$\mathrm{same\_sign}(u_0,u_1,u_2)$}
  \State $eb \gets \mathrm{MAX}\!\left( eb,\; |u_2| \right)$
\EndIf
\If{$\mathrm{same\_sign}(v_0,v_1,v_2)$}
  \State $eb \gets \mathrm{MAX}\!\left( eb,\; |v_2| \right)$
\EndIf
\State \Return $eb$
\end{algorithmic}
\end{algorithm}

\section{Problem Formulation and Guarantees}
\label{sec:formulation}

In this section, we formally define the problem, specify the constraints that guide our compression design, and establish the theoretical guarantees for preserving critical-point trajectories under lossy compression.

\subsection{Research Problem and Goal}
 
Given a vector field $\mathbf{D}=\{\mathbf{u},\mathbf{v}\}$ and a user-specified error bound $\varepsilon$, a compression method $\mathcal{C}$ transforms $\mathbf{D}$ to the compressed data $\mathcal{C}(\mathbf{D})$, which can be decompressed to $\hat{\mathbf{D}} = \mathcal{D}(\mathcal{C}(\mathbf{D}))$ through a corresponding decompression method $\mathcal{D}$. 
Our goal is to maximize the compression ratio, which is defined by $\texttt{size}(\mathbf{D}) / \texttt{size}(\mathcal{C}(\mathbf{D}))$, while enforcing the following error constraints. 

\paragraph{Pointwise error constraint}
\begin{equation}
\label{eq:constraint}
\|\hat{\mathbf{D}} - \mathbf{D}\|_\infty \le \varepsilon;\nonumber
\end{equation}

\paragraph{Critical-point-trajectory constraint}
For each space-time triangular face $f=\{i,j,k\}$ (including both time-slice faces and cross-time-slab faces), we define a face-level critical-point predicate (hereafter referred to as the face predicate) 
\begin{equation}
\label{eq:pf}
P_f(\mathbf{d}_i,\mathbf{d}_j,\mathbf{d}_k)
:= \mathbb{I} \big(\mathbf{0}\in \mathrm{conv}\{\mathbf{d}_i,\mathbf{d}_j,\mathbf{d}_k\} \big).\nonumber
\end{equation}
We require
\begin{equation}
\label{eq:face-invariance}
\forall f:\quad
P_f(\widehat{\mathbf{d}}_i,\widehat{\mathbf{d}}_j,\widehat{\mathbf{d}}_k)
\ =\
P_f(\mathbf{d}_i,\mathbf{d}_j,\mathbf{d}_k).
\end{equation}
Under piecewise-linear interpolation and general position (or SoS), the face-level invariance
\eqref{eq:face-invariance} is equivalent to preserving the set of critical-point trajectories:
\begin{equation}
\label{eq:traj}
\mathcal{T}(\widehat{\mathbf{D}})\ =\ \mathcal{T}(\mathbf{D}),\nonumber
\end{equation}
where $\mathcal{T}(\cdot)$ denotes the \emph{track set}, i.e., the collection of space-time trajectories obtained by
connecting critical points across consecutive timesteps through their face crossings in the space-time mesh. 

\subsection{Guarantees}
\label{sec:problem:guarantees}

We now provide theoretical guarantees showing that the proposed constraints are sufficient to ensure the preservation of per-time critical points and their temporal trajectories.

\begin{lemma}[Zero-set in a tetrahedron]
\label{lem:line}
Under piecewise-linear interpolation in $x$--$y$--$t$, the set $\{\mathbf{V}=0\}$ restricted to any tetrahedron $\tau$ is either empty or a straight line segment whose endpoints lie on two distinct faces of $\tau$. In general position, intersections with the boundary are either $0$ or $2$ points.
\end{lemma}

\begin{theorem}[Per-time critical point preservation]
\label{thm:per-frame}
If the face-level critical-point predicate on every triangular face of the space-time mesh is invariant under compression, then for every time slice $t_k$, the critical point set on $D\times\{t_k\}$ is unchanged (same cardinality and locations via piecewise-linear barycentric mapping).
\end{theorem}
\begin{proof}[Proof Sketch]
A critical point at time $t_k$ is precisely an intersection of the zero-set with a time-slice face lying in $D\times\{t_k\}$; by Lemma~\ref{lem:line}, such intersections are exactly the face crossings decided by the predicate in Sec.~\ref{sec:problem:cp-detection}. If the predicate outcome on all time-slice faces is preserved, the set of intersection points on each slice is identical; their in-face barycentric coordinates (hence geometric positions) are preserved by construction.
\end{proof}

\begin{theorem}[Track equivalence]
\label{thm:track}
Assume general position so that in every tetrahedron the zero-set is either empty or a single segment with endpoints on two faces. If face-level critical point decisions are preserved on all faces, then the adjacency of segments across shared faces is identical before and after compression, hence the reconstructed track graph (number of tracks, their connectivity, and combinatorial structure) is unchanged.
\end{theorem}
\begin{proof}[Proof Sketch]
Within a tetrahedron, the two boundary intersections uniquely determine a segment; thus, preserving which faces are intersected ensures that the same segment is retained. Across tetrahedra, two segments connect if and only if they share the same intersection point on a common face. Since both the face predicate and the corresponding barycentric intersection on that face are preserved, the segment connectivity remains unchanged, yielding an isomorphic track graph. 
\end{proof}

According to Theorem~\ref{thm:per-frame} and Theorem~\ref{thm:track}, 
we can conclude that preserving all per-time critical points as well as all cross-time-slab critical points constitutes a sufficient condition for ensuring the preservation of critical-point trajectories.

\section{Design Overview}
\label{sec:method}

This section provides an overview of our framework and outlines how the proposed method preserves critical-point trajectories during lossy compression.

\begin{figure}[ht]
\centering
\vspace{-1.5em}
\includegraphics[width=1.0\columnwidth]{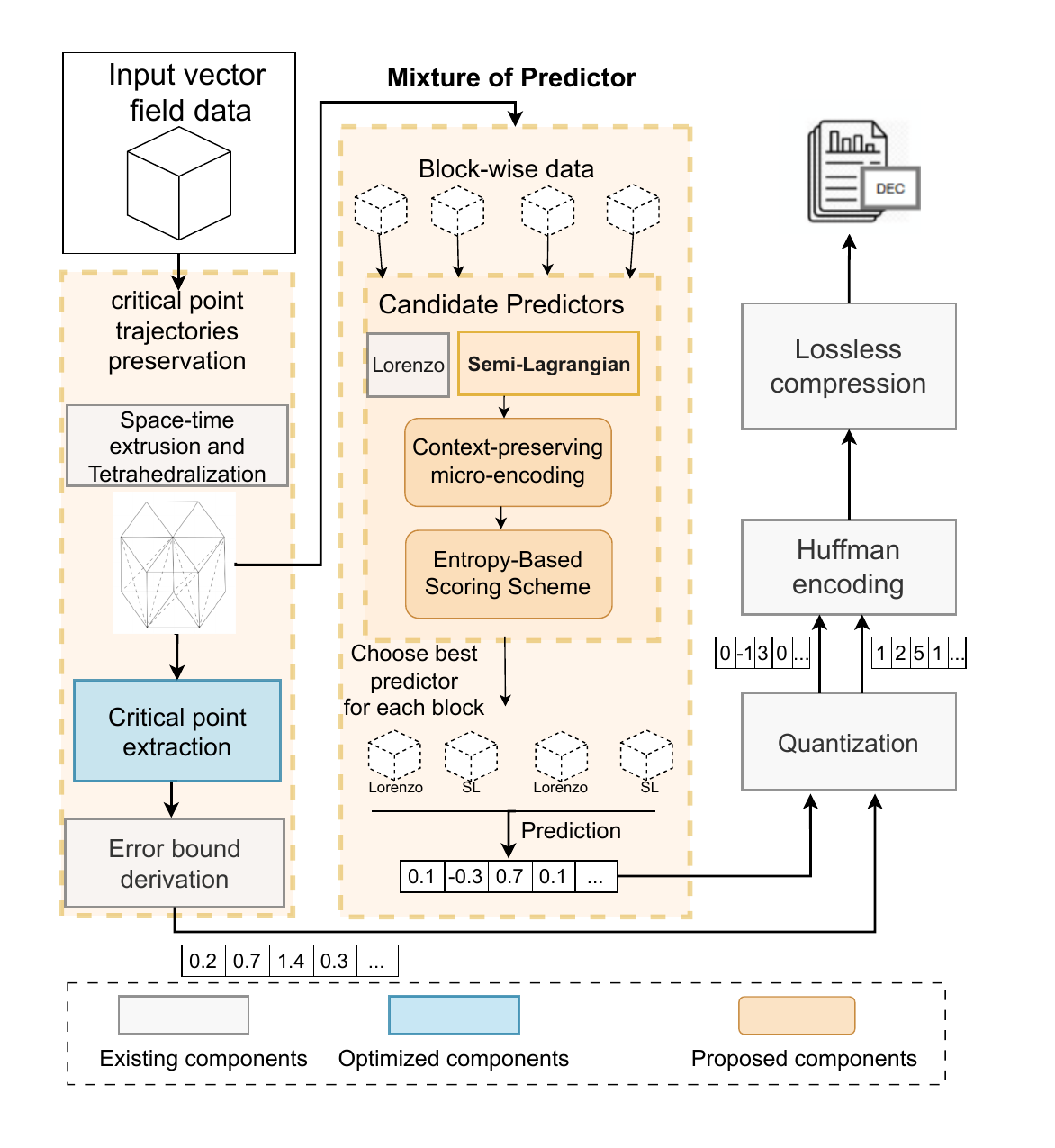}
\vspace{-3em}
\caption{Overview of the proposed framework.}
\label{fig:framework}
\end{figure}

{
\newcommand*{\tikzmk}[1]{\tikz[remember picture,overlay,] \node (#1) {};\ignorespaces}
\begin{algorithm}[htb!]
\caption{Lossy compression for 2D time-varying vector fields with critical-point trajectory preservation} \label{alg:compress2D}
\footnotesize
\renewcommand{\algorithmicrequire}{\textbf{Input:}}
\renewcommand{\algorithmicensure}{\textbf{Output:}}
\renewcommand{\algorithmiccomment}[1]{\hfill\textcolor{gray}{// #1}}
\begin{algorithmic} [1]
    \Require 2D Temporal Vector Fields $\{\mathbf{d}\} = \{\mathbf{U}, \mathbf{V}\}$ of size $N=H\!\times\!W\!\times\!T$; error bound $\epsilon$.
    \Ensure Compressed bytes.
        \State Allocate $U_{\mathrm{fp}},V_{\mathrm{fp}}\in\mathbb{Z}^{N}$ (int64)
        \State $\texttt{scale},U_{\mathrm{fp}},V_{\mathrm{fp}} \gets \texttt{Convert\_FixedPoint}(\mathbf{U}, \mathbf{V})$
        \State $\{\boldsymbol{\tau}\} \gets \texttt{Tetrahedralization}(\mathbf{U},\mathbf{V})$
        \State Initialize $\mathcal{C}$ 
        \For{$tet \in \boldsymbol{\tau}$}
            \For{$f \in tet$} 
                \State $\mathcal{C}.\text{insert}(\texttt{SoS\_robust\_test}(f))$ \Comment{Alg.\ref{alg:robust_cp}}
            \EndFor
        \EndFor
        \State $\epsilon' = \epsilon * \texttt{scale}$
        \State $Q_{\xi},\, Q_{d},\, \mathcal{L}_{d} \gets [\,]$ \Comment{initialize quantization for error bounds, quantization for data, and unpredictable data}

        \For{$t=0,\dots,T-1$}
            \For{$i=0,\dots,H-1$}
                \For{$j=0,\dots,W-1$}
                \State $v \gets VID(t,i,j)$
                    \State $\xi^v \gets \epsilon'$
                    \State $\xi^v, \mathcal{L}_d \gets \texttt{derive\_min\_eb}(v,U_\mathrm{fp},V_\mathrm{fp},\mathcal{C},\epsilon')$ \Comment{Alg.~\ref{alg:derive_eb}}
                    \State  $Q_\xi^v \gets \texttt{quantization}(\xi^v)$
                    \State $Q_d^v,\mathcal{L}_d \gets \texttt{Prediction\_and\_quant}(\mathbf{d}^v,\mathbf{\hat{d}}^v)$
                \EndFor
            \EndFor
        \EndFor
        \State $\texttt{output\_bytes} \gets \texttt{lossless\_comp}(Q_\xi\,Q_d,\mathcal{L}_d)$
        \State \Return $\texttt{output\_bytes}$
\end{algorithmic}
\end{algorithm}
}

{
\algdef{SE}[SUBALG]{Indent}{EndIndent}{}{\algorithmicend\ }
\algtext*{Indent}
\algtext*{EndIndent}

\newcommand*{\tikzmk}[1]{\tikz[remember picture,overlay,] \node (#1) {};\ignorespaces}
\begin{algorithm}[htb!]
\caption{Derive the minimum error bound for a vertex} \label{alg:derive_eb} \footnotesize
\renewcommand{\algorithmicrequire}{\textbf{Input:}}
\renewcommand{\algorithmicensure}{\textbf{Output:}}
\renewcommand{\algorithmiccomment}[1]{\hfill\textcolor{gray}{// #1}}
\begin{algorithmic} [1]
     \Require Vertex $v=\textsc{VID}(t,i,j)$; fixed-point fields $U_{\mathrm{fp}},V_{\mathrm{fp}}$; precomputed face set $\mathcal{C}$; initial bound $\epsilon'$
    \Ensure per-vertex bound $\xi^{(v)}$
    \State $\xi^{(v)} \gets \{\epsilon'\}$
    \State $\mathit{must\_lossless} \gets \mathrm{false}$
    \For{$c_i \in \mathrm{vertex\_cell}(v)$} \Comment{at most 36 triangular cells}
        \State $\triangle^t_{i,j,k} \gets \mathrm{get\_vertices}(v)$
        \If{$\triangle^t_{i,j,k} \in \mathcal{C}$}
            \State $\mathit{must\_lossless} \gets \mathrm{true}$;\ \textbf{break}
        \Else \State append $\mathrm{derive\_eb}(i,j,k)$ to $\xi^{(v)}$ \Comment{Alg.2}
        \EndIf
    \EndFor
        \If{$\mathit{must\_lossless}$}
            \State $\mathcal{L}_d.\mathrm{push}(v)$
            \State \Return $0$ \Comment{lossless store}
        \Else \State \Return $\min \xi^{(v)}$ \Comment{error bound to preserve trajectories}
        \EndIf

\end{algorithmic}
\end{algorithm}
}

\begin{figure*}[htbp]
    \centering
    \includegraphics[width=\linewidth]{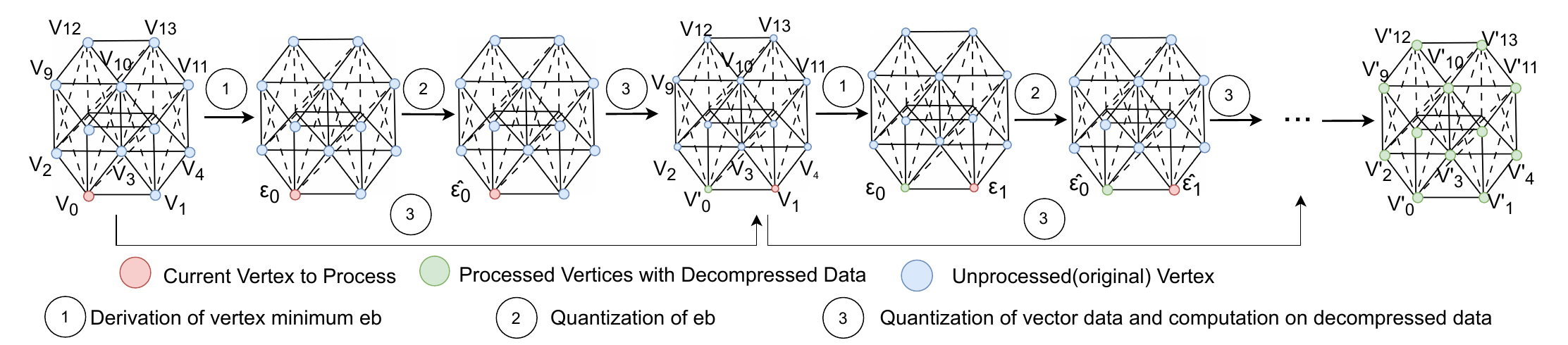}
    \vspace{-2em}
    \caption{An illustrative example of critical-point-trajectory-preserving lossy compression, as described in Algorithm~\ref{alg:compress2D} (lines~10--17).}
    \vspace{-1em}
    \label{fig:2p5d}
\end{figure*}

\subsection{Framework Overview}

As illustrated in Figure~\ref{fig:framework}, our compression pipeline is built on the prior work cpSZ-SoS~\cite{xia2024preserving}, with several key enhancements and modifications. 
The input data are fed into the space-time discretization module, as discussed in Section~\ref{sec:bg:3-split}, which enables our algorithm to handle time-varying 2D data by constructing a consistent spatiotemporal mesh representation. The framework records two components: (1) the critical-point-preserving error bounds associated with each vertex, and (2) the residuals between the original data and their predicted values. These two components are quantized separately and encoded using Huffman coding, followed by an additional lossless compression stage to further improve the overall compression ratio. 

For the critical-point-preserving error bound component, the key difference between our method and prior work lies in how temporal information is handled. By adopting a fixed tetrahedralization strategy, we explicitly consider all triangular faces formed between adjacent time slabs, which allows us to track how critical points enter and exit cells across consecutive time steps. As a result, our method preserves critical-point trajectories in time-varying data, extending previous approaches that were limited to static settings.

For the prediction component, prior work typically employs a single 2D Lorenzo predictor uniformly across the entire dataset. In contrast, our contributions are twofold. First, we introduce a novel semi-Lagrangian predictor, which effectively exploits flow coherence in advection-dominated regions and leads to improved compression performance. Second, recognizing that different regions of a dataset may exhibit distinct dominant characteristics, such as advection-dominated behavior or more locally smooth and diffusion-dominated behavior, we leverage a Mixture of Predictors (MoP) module. This module partitions the data into small blocks, evaluates the prediction performance of multiple predictors on each block, and selects the most suitable predictor accordingly. Further details of the prediction design and the MoP strategy are presented in Section~\ref{sec:mop}.

We present our critical-point-trajectory-preserving lossy compression in Algorithm~\ref{alg:compress2D}. We first convert the input data into a fixed-point representation so that the SoS method can be applied. The data are then tetrahedralized and partitioned according to a predefined strategy. Next, a critical point test is performed on all triangular faces of each tetrahedron, and the results are recorded. 
For each vertex, we derive an error bound that guarantees the preservation of critical points based on its incident triangular faces (up to 36 faces), and quantize the vertex value accordingly. Subsequently, prediction and quantization are applied to the vector data at this vertex. Then, we compute the decompressed data and replace the original value with it on the fly, ensuring that subsequent processing operates on previously decompressed results. In this way, the compression process is performed sequentially, and the processing of each data point explicitly depends on the decompressed values of preceding points. Finally, all quantized error bounds, quantized vector data, and unpredictable data that must be stored losslessly are packaged together and passed to a lossless compressor for final encoding.

Figure~\ref{fig:2p5d} provides a concrete example of this procedure, illustrating how each data point is processed in sequence and subsequently replaced by its decompressed value. The figure highlights the step-by-step dependency in the compression pipeline, where the reconstructed result of a previously processed point is immediately used as input for the processing of the next point, thereby ensuring consistency throughout the compression process.

\subsection{Time Complexity Analysis}
Denote $N$ as the total number of data points, and we analyze the complexity of our compression algorithm as follows. Our algorithm consists of three stages.
In the first stage, we perform a robust critical point test on all triangles incident to each vertex.
Each test has a time complexity of $O(m)$, where $m$ denotes the dimensionality ($m=2$ for the 2D datasets).
According to our prism decomposition strategy, each vertex is incident to at most 36 triangles.
Therefore, the complexity of this stage is $O(\alpha m N)$, where $\alpha = 36$ denotes the maximal number of triangular faces associated with a data point.
Since both $\alpha$ and $m$ are constants, this stage has an overall time complexity of $O(N)$.
As a variant of the SZ compressor family, prior work shows this stage has complexity $O(N)$~\cite{sz16}.
The third stage derives error bounds for preserving critical-point trajectories.
Similar to the first stage, it requires per-vertex analysis over incident triangles and therefore has complexity $O(\alpha N)=O(N)$.
Overall, the total time complexity of our algorithm is $O(N)$.
\section{Block-Wise Adaptive Mixture of Predictors}
\label{sec:mop}

Prediction plays a crucial role in improving compression efficiency, as accurate prediction effectively reduces the magnitude of residuals to be encoded, leading to a more compact entropy distribution and thus higher compression ratios. In particular, prior research~\cite{sz18,liu2025solarzip} has shown that block-based predictive coding effectively adapts to local data heterogeneity, which is crucial for scientific datasets exhibiting complex spatial and temporal patterns.

We focus on compressing 2D time-varying vector field data $\mathbf{v}(x,y,t) = (u,v)$. Numerical simulation data often exhibit strong local physical heterogeneity. In different subdomains, the dominant physical mechanisms such as advection, convection, nonlinear transport, and diffusion can vary significantly. Most state-of-the-art scientific lossy compressors~\cite{sz17,sz16,sz18,liang2018efficient} adopt the first-order Lorenzo predictor as their prediction model. We refer readers to~\cite{ibarria2003out} for further details. This model performs well in locally smooth or diffusion-dominated regions due to its trilinear reproduction property, but it produces large residuals in regions with significant translational or transport structures (for example, vortex-core migration or stripe advection), which reduces compression efficiency.

To address this limitation, we propose a \textbf{block-wise adaptive mixture of predictors (MoP)} with two candidates: \textbf{first-order 3D-Lorenzo (3DL)} predictor and  \textbf{semi-Lagrangian (SL)} predictor tailored for advection-dominated regions. 
We partition each frame at time $t$ into non-overlapping blocks $B$ of size $B_x\times B_y$. For each spatial block, we evaluate both predictors with a light-weight, context-preserving look-ahead, then select the mode that minimizes the estimated rate; this reduces residual entropy and improves overall compression ratio while strictly preserving the user-prescribed error bound.

\vspace{-0.5\baselineskip}
\subsection{Semi-Lagrangian Predictor}
\vspace{-0.2\baselineskip}

Let the spatial steps be $\Delta x$ and $\Delta y$, and the time step be $\Delta t$; these parameters are provided in the data metadata or descriptive information. Let $S$ denote the fixed-point scaling factor, such that the data are stored as integers scaled by $S$. 
Define per-axis \textbf{Courant–Friedrichs–Lewy (CFL)}-like factors, 
\[
\mathrm{CFL}_x=\frac{\Delta t}{\Delta x}\cdot\frac{1}{S},\qquad
\mathrm{CFL}_y=\frac{\Delta t}{\Delta y}\cdot\frac{1}{S}.
\]
At grid point $(i,j)$ we form a backward displacement using the previous velocity field $(u_{t-1},v_{t-1})$ and sample the previous frame at the departure point with bilinear interpolation:
\begin{align}
\big(i^\star,j^\star\big)
&=\Big(i - u_{t-1}(i,j)\,\mathrm{CFL}_x,\;
       j - v_{t-1}(i,j)\,\mathrm{CFL}_y\Big) \label{eq:sl-dep}\\
\widehat{f}^{\text{SL}}_{i,j,t}
&=\mathcal{I}_{\text{bil}}\!\left[f_{t-1}\right]\!\big(i^\star,j^\star\big) \label{eq:sl-sample}
\end{align}
where $\mathcal{I}_{\text{bil}}$ is standard bilinear interpolation at $(i_f,j_f)$ with clamping at domain boundaries.

Local flow velocities can vary significantly across blocks. When the velocity in the current block is large, the resulting displacement may cross multiple cells or even go out of bounds. This can lead to unstable time integration or large numerical errors. \revise{Therefore, we design two strategies to handle stable and unstable displacement regimes, respectively.}

We first estimate a dimensionless local displacement magnitude $d_\infty$ using the previous timestamp value. 
\[
d_\infty=\max\!\Big(\big|u_{t-1}(i,j)\big|\,\mathrm{CFL}_x,\;\big|v_{t-1}(i,j)\big|\,\mathrm{CFL}_y\Big).
\]

\noindent\textbf{RK2 (mid-point).}
When $d_\infty$ is smaller than a prescribed threshold $d_{\max}$ (set as 2 pixels), the displacement remains within a stable regime. In this case, we apply the second-order Runge–Kutta (mid-point) method to improve temporal accuracy:
{\small
\begin{align}
    &(i_{\frac{1}{2}},j_{\frac{1}{2}}) = \Big(i-\tfrac{1}{2}u_{t-1}(i,j)\,\mathrm{CFL}_x,\;
    j-\tfrac{1}{2}v_{t-1}(i,j)\,\mathrm{CFL}_y\Big), \label{eq:rk2-ij_half}\\
    &u_{t-1}(i_{\frac{1}{2}},j_{\frac{1}{2}}) = \mathcal{I}_{\text{bil}}\!\left[u_{t-1}\right]\!\big(i_{\frac{1}{2}},j_{\frac{1}{2}})\nonumber,
    v_{t-1}(i_{\frac{1}{2}},j_{\frac{1}{2}}) = \mathcal{I}_{\text{bil}}\!\left[v_{t-1}\right]\!\big(i_{\frac{1}{2}},j_{\frac{1}{2}})\nonumber\\
    &(i^\star,j^\star) = \Big(i- u_{t-1}(i_{\frac{1}{2}},j_{\frac{1}{2}})\,\mathrm{CFL}_x,\;
    j- v_{t-1}(i_{\frac{1}{2}},j_{\frac{1}{2}})\,\mathrm{CFL}_y\Big) \label{eq:rk2_ij_star}
\end{align}
}

\noindent\textbf{Adaptive substepping.}
When $d_\infty > d_{\max}$, a single-step integration may become unstable. To ensure robustness, we adopt an adaptive substepping strategy by splitting the displacement into $N = \min\!\bigl(\lceil d_\infty / d_{\max} \rceil,\, N_{\max}\bigr)$ explicit substeps, where $N_{\max}=32$ by default. Each substep is given by:
{\small
\begin{align}
(i_{s+1},j_{s+1})
&= (i_s,j_s)
   - \Big(
       u_{t-1}(i_s,j_s)\,\tfrac{\mathrm{CFL}_x}{N},\;
       v_{t-1}(i_s,j_s)\,\tfrac{\mathrm{CFL}_y}{N}
     \Big) \nonumber\\
& \text{where }s = 0,\ldots,N-1.
\label{eq:substeps}
\end{align}}
\noindent~with clamping to the data domain at each substep, and finally we have$(i^\star,j^\star)=(i_N,j_N)$. 
After we obtain $(i^\star,j^\star)$, use~\eqref{eq:sl-sample} to predict the $u$ or $v$ component.

\begin{figure}[htbp]
\centering
\includegraphics[width=1.0\columnwidth]{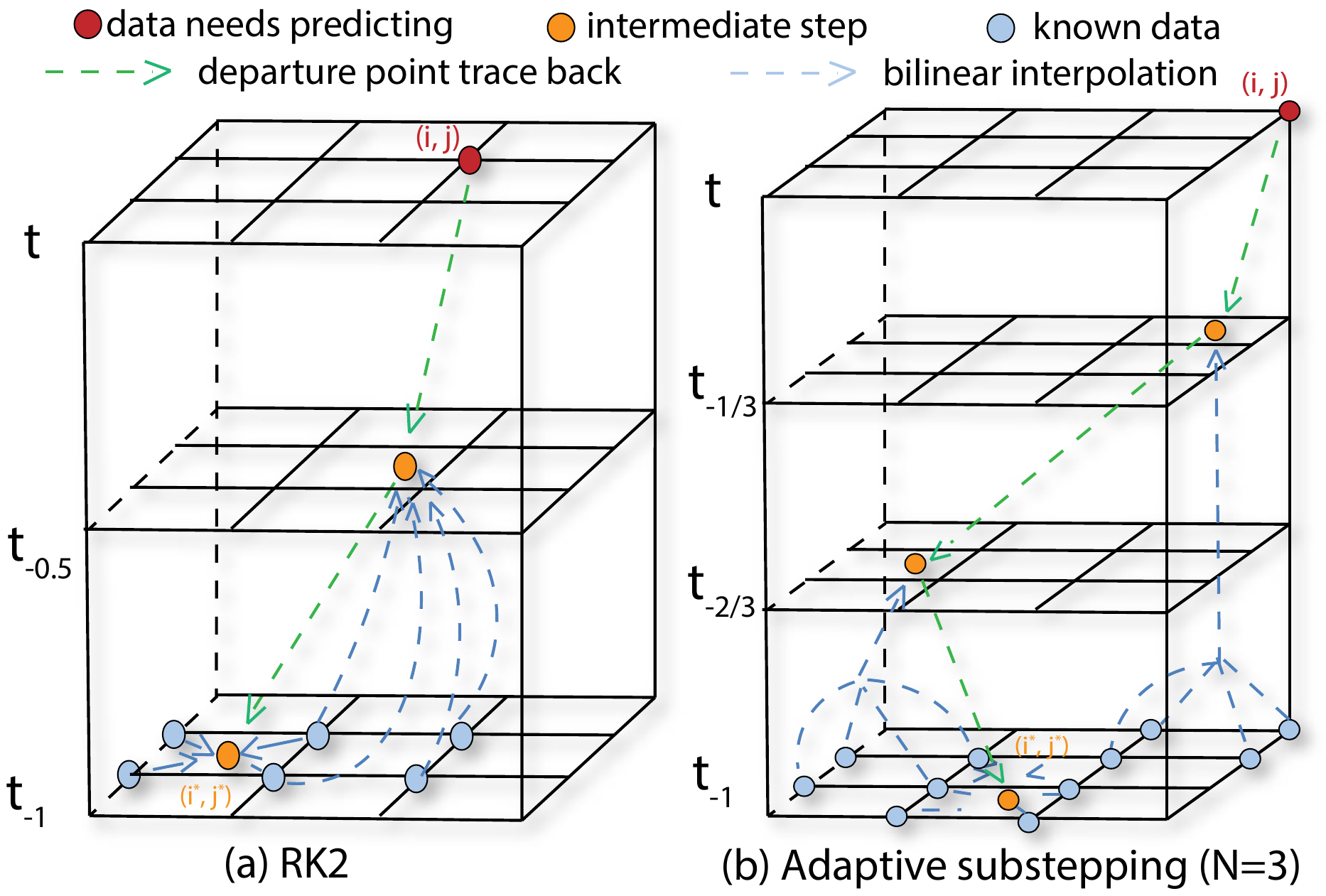}
\vspace{-\baselineskip}
\caption{RK2 and adaptive substepping method in semi-Lagrangian predictor. The green dashed arrows indicate the departure point traceback (Eq.~\ref{eq:sl-dep}), while the light blue dashed arrows represent the bilinear interpolation (Eq.~\ref{eq:sl-sample}).}\label{fig:sl}
\end{figure}

Figure~\ref{fig:sl} provides a schematic illustration of the two proposed prediction methods. The red point indicates the target location to be predicted. We first compute the coordinates of all intermediate steps (yellow points) using \eqref{eq:rk2-ij_half} or \eqref{eq:substeps}. Then, for each intermediate step, we perform bilinear interpolation of the components $u$ and $v$ at time $t-1$ from its four surrounding points to obtain the local velocity $(u,v)$ at the current intermediate position. This process is repeated iteratively until the final position $(i^\star, j^\star)$ is reached, where we again perform bilinear interpolation of the four neighboring points at time $t-1$ to obtain the predicted value.

\subsection{Data Sampling and Estimation}

\noindent\textbf{Context-preserving micro-encoding.} For each block $B$ and each mode $p$, we traverse every data point in $B$, compute prediction $\hat f^{(p)}$, form residual $r=x-\hat f^{(p)}$, and execute the quantization and dequantization process identical to that of the compressor. This procedure requires an error bound $\epsilon'$ that preserves critical points. In practice, $\epsilon'$ is typically tighter than the user-specified tolerance $\epsilon$. For optimistic estimation, we use $\epsilon$ in our calculations. After that, we write back the reconstructed values into a candidate buffer so that subsequent data see the same neighborhood as in true encoding. This guarantees the scoring respects the causal context of the actual compressor. To reduce scoring overhead, we collect statistics only on a subsampled grid with stride $n$ along both directions.

\noindent\textbf{Score computation.}
Let $h_p[k]$ be the histogram over all subsampled and valid quantized residuals from both components in $B$ under mode $p$, and let $n_p=\sum_k h_p[k]$ be the valid-sample count.
Let $\mathcal{L}_p$ be the number of overflow/failure events among the subsampled sites (these data needs to be compressed losslessly).
We compute the zero-order entropy (bits per sample) as
\begin{align}
H_0(h_p;n_p) \;=\; -\sum_{k}\frac{h_p[k]}{n_p}\,\log_2\!\Big(\frac{h_p[k]}{n_p}\Big),
\qquad n_p>0.
\label{eq:entropy}
\end{align}

The estimated bits per sample per component for mode $p$ combines entropy, an overflow penalty, and the amortized block-mode metadata:
\begin{align}
\widehat{R}_p(B)
\;=\;
H_0(h_p;n_p)
\;+\; \lambda\,\widehat{\xi}_p
\;+\; R_{\mathrm{meta}},\;\;
\widehat{\xi}_p \;=\; \frac{\mathcal{L}_p}{\,n_p+\mathcal{L}_p\,}.\nonumber
\end{align}

Here, $\lambda$ denotes a fixed penalty measured in bits. We set it to half the data type size, assuming that subsequent lossless compression achieves an average compression ratio of $2\times$~\cite{lindstrom2017error,son2014data}. Since only two modes are considered, the metadata overhead is one bit per block. 
We select the mode with the lower estimated score, subject to a relative-improvement criterion. Specifically, if the improvement of $\widehat{R}_p^{SL}(B)$ exceeds a switching threshold (set to $0.3\%$), we adopt the SL mode for that block; otherwise, we revert to 3DL to avoid mode thrashing caused by noisy estimates. 

\subsection{Hyperparameter Selection}
\begin{wrapfigure}[12]{r}{1.8in}
    \vspace{-0.1in}
    \includegraphics[width=\linewidth]{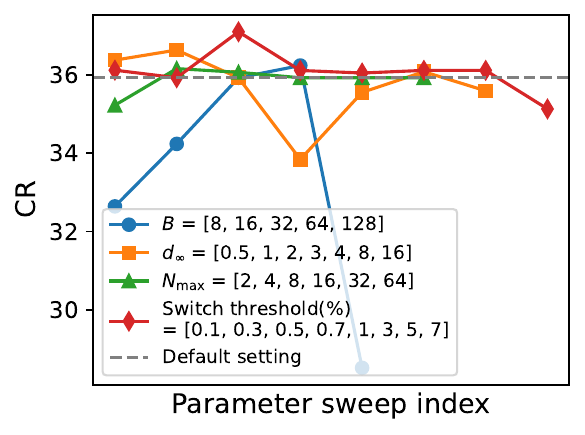}
    \vspace{-2em}
    \caption{\revise{Compression ratio with eb=5 on the SCF dataset obtained by sweeping individual hyperparameters while keeping all others at their default settings.}}
    \label{fig:parameter_sweep} 
    \vspace{-2em}
\end{wrapfigure}
Figure \ref{fig:parameter_sweep} illustrates the sensitivity of the compression ratio to key hyperparameters for MoP method. Block size $B$ depends on the spatial resolution of each time slice and the intrinsic complexity of the underlying flow structures. Smaller blocks can better capture fine-grained features in complex data but increase metadata overhead due to a larger number of blocks, which may reduce the overall compression ratio. A larger $d_{\text{max}}$ causes more data to use the RK2 instead of the adaptive substepping method when computing $(i^*,j^*)$. Since RK2 is less accurate, overly large values may degrade prediction quality and reduce compression efficiency. 

$N_{\text{max}}$ limits the maximum number of iterations in the adaptive substepping method, preventing excessive computational cost when large displacements require repeated step subdivision. When the SL method is the dominant predictor, an excessively large switch threshold may force blocks that should be predicted by SL to be predicted using 3DL instead, thereby resulting in a degradation in compression ratio. In practice, these parameters introduce typical accuracy–efficiency trade-offs; however, we observe that the method is not sensitive to their exact values within a reasonable range. Accordingly, we use fixed settings that provide stable performance across datasets.

\section{Evaluation}
\label{sec:evaluation}

In this section, we present a comprehensive evaluation of our proposed compression framework and compare it against state-of-the-art compressors.
\subsection{Experimental Setup and Datasets}

\begin{table}[ht]
\vspace{-4mm}
\centering
\caption{Benchmark datasets}
\label{tab:data}
\normalsize
\vspace{-1mm}
\resizebox{\columnwidth}{!}{%
{\color{black}
\begin{tabular}{|c|c|c|c|c|}
\hline
\multicolumn{1}{|l|}{\textbf{Dataset}} & \textbf{Dim(H,W,T)} & \textbf{Range}          & \textbf{\#CP} & \textbf{Size(GB)} \\ \hline
SCF                                    & $(450,200,500) $    & [-65.21, 54.11] & 348769        & 0.33     \\ \hline
CFVKV                                  & $(640,80,1501)$     & [-0.93, 1.83{]}   & 22893         & 0.57     \\ \hline
HCBA                                   & $(150,450,2001)$    & [-1.25, 1.02]   & 148501        & 1.01     \\ \hline
FS                                     & $(512,512,1001)$    & [-0.90, 0.83]   & 26861         & 1.95     \\ \hline
DT                                     & $(2048,2048,1001)$  & [-0.99, 0.96]   & 361086        & 31.28    \\ \hline
\end{tabular}
}
}
\end{table}

We conduct our experiments using five scientific datasets listed in Table~\ref{tab:data}, where \#CP denotes the total number of critical points located at time $t$, and within a time slab. All datasets~\cite{Jung93,gerrisflowsolver,Guenther17,Jakob20,Basilisk} are stored as single-precision floating-point values with the following detailed information.

\begin{itemize}
    \item \textbf{Synthetic Cylinder Flow (SCF)}: A synthetic vector field models the generation of a simple von Krman vortex street.
    
    \item {\bfseries \textrm{Cylinder Flow with von Kármán Vortex (CFVKV)}}: simulation of a viscous 2D flow around a cylinder.
    
    \item \textbf{Heated Cylinder with Boussinesq Approximation (HCBA):} A simulation of a 2D flow generated by a heated cylinder, using Boussinesq approximation.
    \item \textbf{Fluid Simulation Ensemble (FS)}: An ensemble simulation of 8000 unsteady 2D flows, with turbulence becoming more prominent in the later snapshots. Without loss of generality, we randomly pick the $6666^{th}$ flow dataset for evaluation.
    \item \textbf{Decaying Turbulence (DT)}: A 2D decaying turbulence dataset generated using the Basilisk solver with random band-limited initialization. 
    
\end{itemize}

We conduct the experiments on a medium-sized HPC cluster~\cite{mcc}. Each node is equipped with AMD EPYC 7702 CPUs and 512 GB of DDR4 memory, running CentOS~8.4 and GCC~9.3.

\subsection{Baselines}
To the best of our knowledge, apart from lossless compression, no existing lossy compressor is capable of preserving the critical-point trajectories in time-varying 2D vector field data. Since lossless compressors such as GZIP, ZSTD, and FPZIP inherently guarantee full topological fidelity, we use them as upper-bound baselines for comparison. In addition, we evaluate three representative lossy compressors, that is, ZFP, SZ3, and cpSZ-SoS, to assess how lossy techniques perform in preserving critical-point trajectories. All evaluated compressors have $O(N)$ time complexity.
 
\subsection{Metrics}
Our goal is to ensure that the critical points remain consistent before and after compression and to preserve the temporal coherence of their trajectories.
We employ the following metrics for quantitative evaluation:
\begin{itemize}
  \item \textbf{Compression Ratio (CR):} 
  CR measures the data reduction capability, as defined in Section~\ref{sec:formulation}.
  
  \item \textbf{Peak Signal-to-Noise Ratio (PSNR):}
  PSNR evaluates overall field fidelity in the value domain, which is defined as
$\text{PSNR} = 20 \cdot \log_{10}(\text{data range}) - 10 \cdot \log_{10}(\text{MSE})$. 

  \item \textbf{False Cases of Critical Points:}
  This metric quantifies the incorrect detection or correspondence of critical points, 
  including both false positives (spurious critical points) and false negatives (missing critical points).
  It is measured at two temporal scales:
  (a) spatial level at time $t$ ($FC_t$); (b) space-time level at time slab $[t_k,t_{k+1}]$ ($FC_s$).
  \item \textbf{Number of Trajectories:}
  The number of critical-point trajectories formed by connecting
critical points detected at each time step $t$ through their
correspondences across consecutive time slabs.
\end{itemize}

\begin{figure*}[hbtp]
\includegraphics[width=\linewidth]{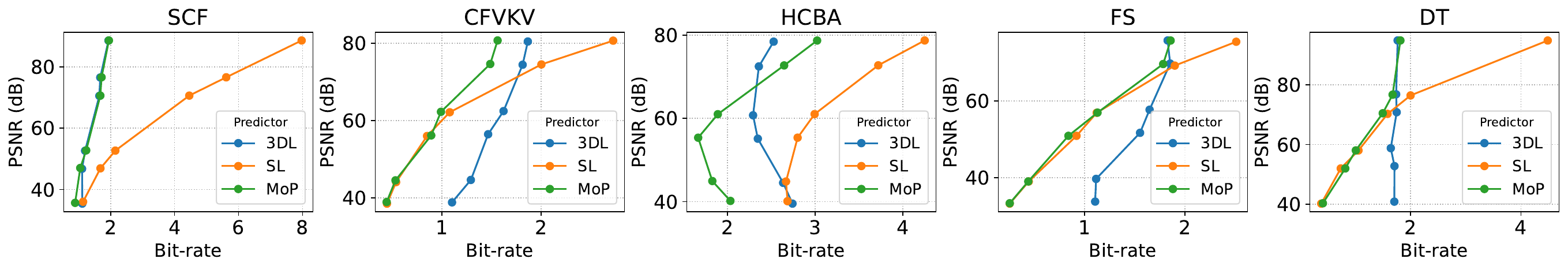}
\vspace{-2em}
\caption{Rate-distortion curves of the three methods on five datasets under varying error bounds. The x-axis denotes the bit rate (data type size in bits divided by the compression ratio), and the y-axis denotes PSNR. CFVKV, FS, and HCBA use [5e-4, 1e-3, 5e-3, 1e-2, 5e-2, 1e-1]; SCF uses [1e-2, 5e-2, 1e-1, 0.5, 1, 5]; and DT uses [1e-4, 5e-4, 1e-3, 5e-3, 1e-2, 5e-2]. At the same PSNR, a lower bit rate indicates better compression performance.}\label{fig:rate_distortion}
\vspace{-2mm}
\end{figure*}

\begin{figure}[hbtp]
\centering
\includegraphics[width=1\columnwidth]{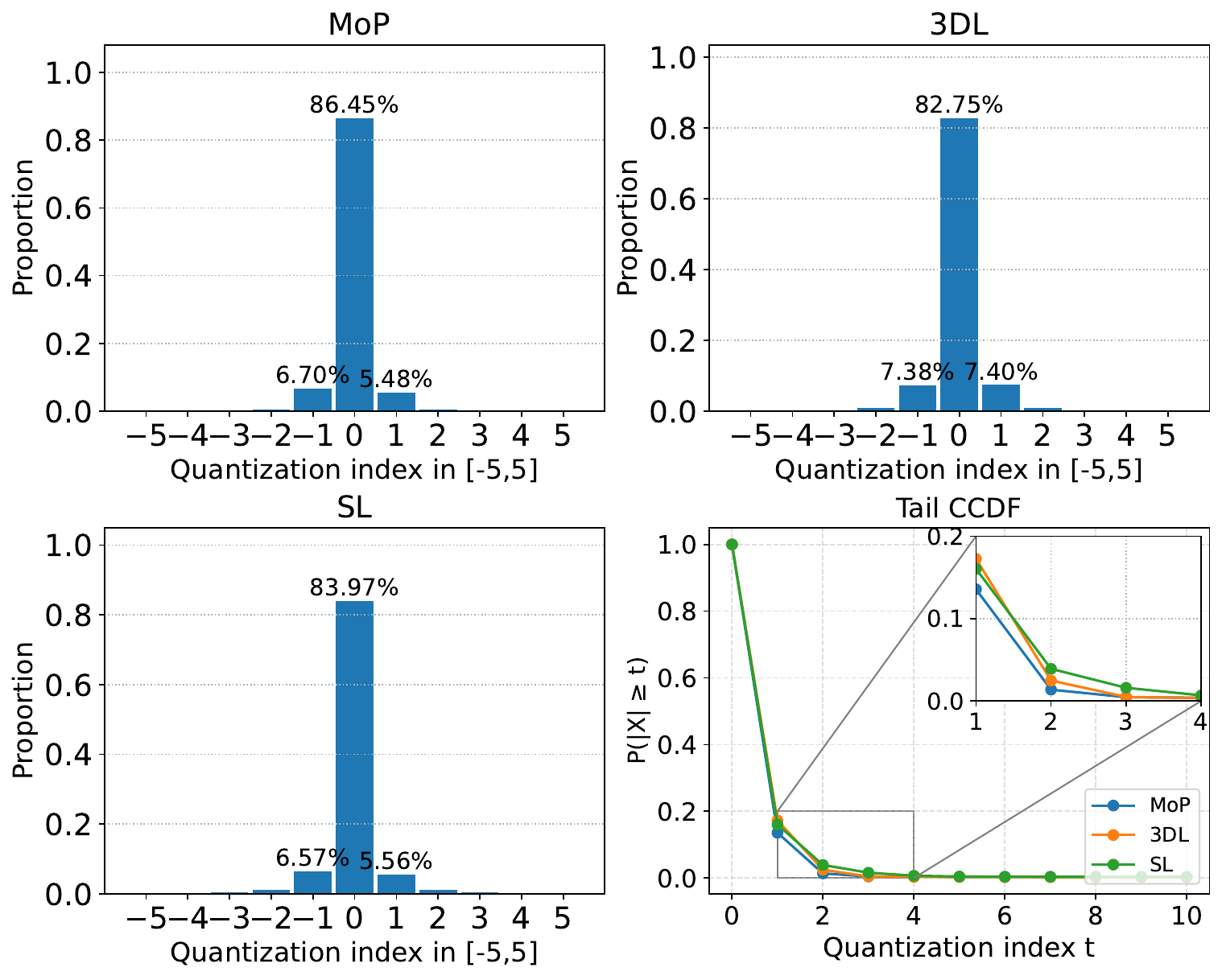}
\vspace{-5mm}
\caption{Probability mass function (PMF) and complementary cumulative distribution function (CCDF) of the quantization index under three methods on the SCF dataset at eb=5.}\label{fig:pmf}
\vspace{-3mm}
\end{figure}

\begin{figure}[hbtp]
\vspace{-0mm}
\centering
\includegraphics[width=1\columnwidth]{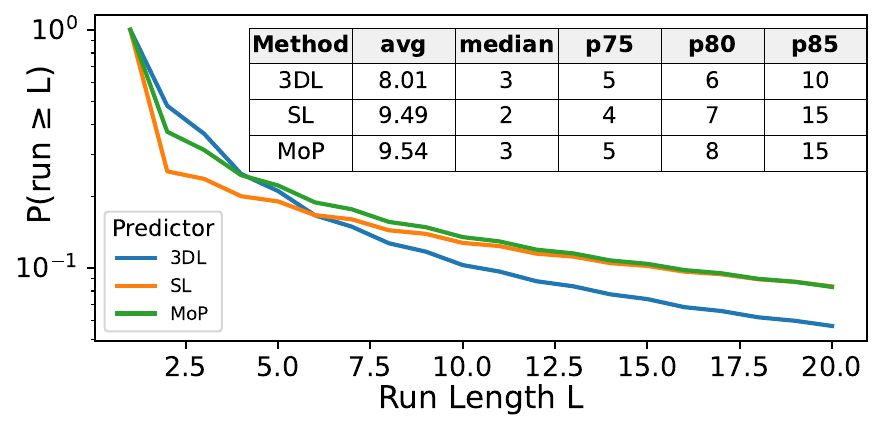}
\vspace{-8mm}
\caption{CCDF of run length between 0 and 20 on SCF dataset at eb=5.}\label{fig:ccdf_run}
\end{figure}

\subsection{Ablation Study on Rate-Distortion Performance}
Figure~\ref{fig:rate_distortion} presents the rate-distortion curves of our three methods across five datasets, explicitly illustrating the sensitivity to different error bounds. For each dataset, the error bounds are selected relative to the data range.   
We observe that, under the same PSNR, the SL method performs exceptionally well in terms of compression ratio on the CFVKV and FS datasets when the error bounds are relatively large, while the 3DL method achieves better compression ratios on the SCF and FS datasets. In contrast, the MoP method effectively combines the strengths of both approaches, achieving a compression ratio that either surpasses or closely matches the best of the other two methods in most cases.

\subsection{Encoding Efficiency}
\label{subsec:efficiency}

In our pipeline, two factors dominate the final bit rate: the fixed error bounds after quantization, $\{Q_\xi\}$, which guarantee critical-point tracking, and the distribution of the quantization indices, $\{Q_d\}$. The error bound is analytically derived from the requirement and is therefore not a tunable lever for further compression. Consequently, our effort focuses on shaping the index distribution via better prediction so that the subsequent entropy and lossless coders operate more efficiently. The rate gains of MoP come from two layers: (1) \textbf{Zero-order statistics.} Let $p(k)$ denote the PMF of the folded quantization index (central bins correspond to $|q|\le 1$ under our toward-zero mapping). The ideal per-sample cost of Huffman coding is $H_0(p)=-\sum_k p(k)\log_2 p(k)$. By concentrating probability mass near the center, MoP reduces $H_0(p)$ compared with 3DL and the SL predictor. In Figure~\ref{fig:pmf}, the MoP method exhibits a markedly sharper central peak than 3DL and SL. Meanwhile, the tail CCDF illustrates the relative proportions of data in both tails. It can be seen that MoP exhibits a smaller proportion beyond the range $[-3,3]$, suggesting a higher concentration of data around the center which directly translates into shorter Huffman codewords and a lower first-order rate. (2) \textbf{Higher-order structure.}  After entropy coding, the bytestream is fed to a backend lossless compressor (ZSTD), which exploits repetition and context (LZ77-style matching) beyond zero-order statistics. Figure~\ref{fig:ccdf_run} shows the complementary cumulative distribution function (CCDF) of the run lengths for the center symbol obtained from three different predictors: 3DL, SL, and MoP. The x-axis represents the run length $L$, while the y-axis denotes the probability $P(run \ge L)$, i.e., the proportion of continuous run length is greater than or equal to $L$. This figure provides a statistical view of how long the correctly predicted regions persist continuously across space and time. In other words, a fatter tail corresponds to more regular and persistent higher-order structures. According to the figure, we can observe that both SL and MoP are able to capture more longer run-lengths compared with 3DL. In addition, the upper-right panel of the figure shows the corresponding means and multiple percentiles for the three methods, which further confirm this observation. Since ZSTD leverages repeated substrings, a longer zero-run yields longer matches and thus additional savings on top of Huffman, further improving the overall compression ratio.

\begin{figure}[ht]
    \centering
    \includegraphics[width=1\linewidth]{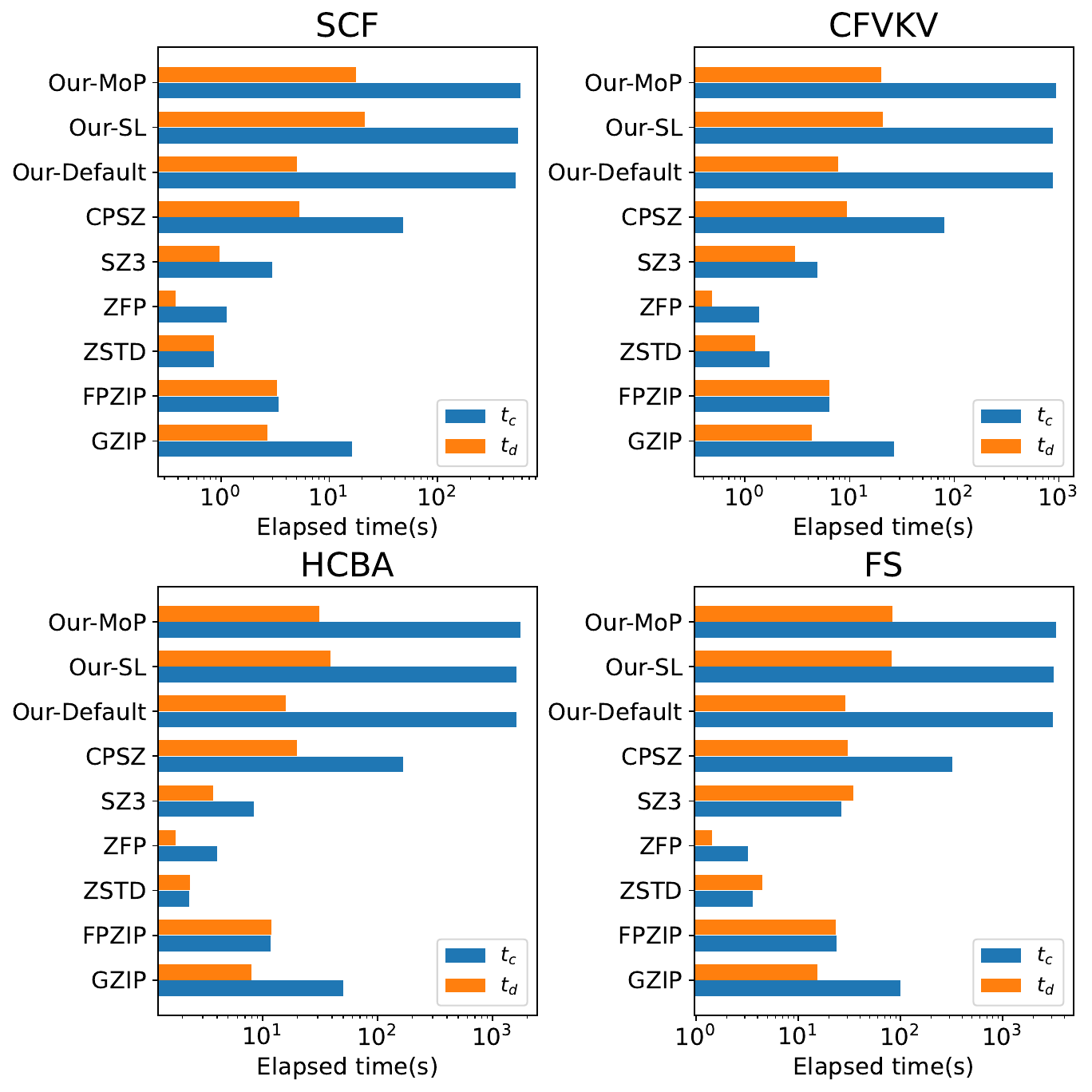}
    \vspace{-4mm}
    \caption{Compression and decompression time ($t_c$ and $t_d$) comparison under four datasets.}
    \label{fig:time_comparison}
    \vspace{-2mm}
\end{figure}

\begin{figure*}[htbp]
    \centering
    \includegraphics[width=\linewidth]
    {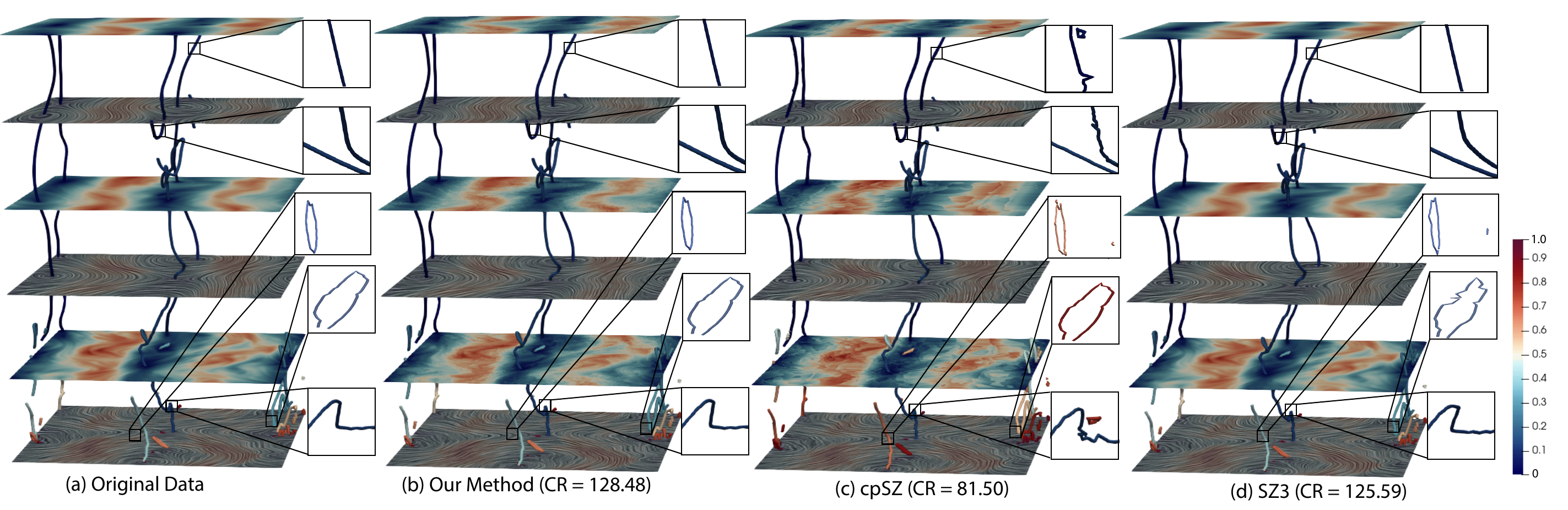}
    \vspace{-2em}
    \caption{Visualization of critical-point trajectories in the FS dataset.}
    \vspace{-2mm}
    \label{fig:FS}
\end{figure*}

\begin{table}[htbp]
\large
\centering
\caption{Quantitative results across all datasets}
\label{tab:all_datasets}
\vspace{-3mm}
\resizebox{\linewidth}{!}{
\begin{tabular}{|l|l|c|c|c|c|c|c|}
\hline
\textbf{Dataset} &
\textbf{Compressors} &
\textbf{Settings} &
\textbf{CR$_{\text{all}}$} &
\textbf{PSNR} &
\textbf{\#FC$_t$} &
\textbf{\#FC$_s$} &
\textbf{\#Traj} \\
\hline

\multirow{9}{*}{SCF}
 & GZIP        & \multirow{3}{*}{/}     & 1.10 & \multirow{3}{*}{/} & \multirow{3}{*}{0} & \multirow{3}{*}{0} & \multirow{3}{*}{999} \\
 & ZSTD        &                        & 1.10 &                    &  &  &  \\
 & FPZIP       &                        & 3.31 &                    &  &  &  \\
\cline{2-8}
 & ZFP         & $\epsilon=0.5$         & 34.74 & 82.63  & 128640 & 347383 & 4467 \\
 & SZ3         & $\epsilon=2$E-4        & 34.18 & 122.71 & 34166  & 89657  & 4399 \\
 & cpSZ(SoS)   & $\epsilon=5$           & 30.90 & 35.51  & 0      & 23486  & 2690 \\
\cline{2-8}
 & \multirow{3}{*}{Ours}  & Default $\epsilon=5$   & 28.71 & 35.43  & \multirow{3}{*}{0} & \multirow{3}{*}{0} & \multirow{3}{*}{999} \\
 &             & SL $\epsilon=5$        & 28.10 & 36.04  &  &  &  \\
 &             & MoP $\epsilon=5$       & 35.93 & 35.68  &  &  &  \\
\hline

\multirow{9}{*}{CFVKV}
 & GZIP        & \multirow{3}{*}{/}     & 1.32 & \multirow{3}{*}{/} & \multirow{3}{*}{0} & \multirow{3}{*}{0} & \multirow{3}{*}{37} \\
 & ZSTD        &   & 1.32 &  & &  &  \\
 & FPZIP       &   & 2.59 &  &  &  &  \\
\cline{2-8}
 & ZFP         & $\epsilon=0.4$         & 64.91 & 56.94  & 30274 & 75668  & 815 \\
 & SZ3         & $\epsilon=6$E-4        & 69.41 & 81.35  & 42116 & 133681 & 3825 \\
 & cpSZ(SoS)   & $\epsilon=0.1$         & 30.71 & 39.31  & 0     & 414    & 39 \\
\cline{2-8}
 & \multirow{3}{*}{Ours}  & Default $\epsilon=0.1$ & 29.04 & 38.84  & \multirow{3}{*}{0} & \multirow{3}{*}{0} & \multirow{3}{*}{37} \\
 &             & SL $\epsilon=0.1$      & 71.82 & 38.51  &  &  &  \\
 &             & MoP $\epsilon=0.1$     & 72.29 & 38.93  &  &  &  \\
\hline

\multirow{9}{*}{HCBA}
 & GZIP        & \multirow{3}{*}{/}     & 1.13 & \multirow{3}{*}{/} & \multirow{3}{*}{0} & \multirow{3}{*}{0} & \multirow{3}{*}{167} \\
 & ZSTD        &  & 1.10 &  &  &  &  \\
 & FPZIP       &  & 2.57 &  &  &  &  \\
\cline{2-8}
 & ZFP         & $\epsilon=1$E-3        & 14.31 & 96.57  & 76681  & 284941 & 8140 \\
 & SZ3         & $\epsilon=4$E-5        & 15.79 & 101.56 & 134907 & 446320 & 12110 \\
 & cpSZ(SoS)   & $\epsilon=5$E-3        & 15.14 & 61.08  & 0      & 9604   & 314 \\
\cline{2-8}
 & \multirow{3}{*}{Ours}  & Default $\epsilon=0.1$ & 11.68 & 39.50  & \multirow{3}{*}{0} & \multirow{3}{*}{0} & \multirow{3}{*}{167} \\
 &             & SL $\epsilon=0.1$      & 11.92 & 40.09  &  &  &  \\
 &             & MoP $\epsilon=0.1$     & 15.74 & 40.15  &  &  &  \\
\hline

\multirow{9}{*}{FS}
 & GZIP        & \multirow{3}{*}{/}     & 1.12 & \multirow{3}{*}{/} & \multirow{3}{*}{0} & \multirow{3}{*}{0} & \multirow{3}{*}{40} \\
 & ZSTD        &  & 1.11 &   &  &  &  \\
 & FPZIP       &  & 2.22 &   &  &  &  \\
\cline{2-8}
 & ZFP         & $\epsilon=4$           & 121.72 & 35.80 & 12362 & 51535 & 107 \\
 & SZ3         & $\epsilon=8$E-4        & 125.59 & 75.51 & 1847  & 8022  & 44 \\
 & cpSZ(SoS)   & $\epsilon=0.1$         & 81.50  & 34.06 & 0     & 4415  & 121 \\
\cline{2-8}
 & \multirow{3}{*}{Ours}   & Default $\epsilon=0.1$ & 28.91 & 33.77 & \multirow{3}{*}{0} & \multirow{3}{*}{0} & \multirow{3}{*}{40} \\
 &             & SL $\epsilon=0.1$      & 124.71 & 33.24 &  &  &  \\
 &             & MoP $\epsilon=0.1$     & 124.48 & 33.34 &  &  &  \\
\hline

\multirow{9}{*}{DT}
 & GZIP        & \multirow{3}{*}{/}  & 1.10 & \multirow{3}{*}{/} & \multirow{3}{*}{0} & \multirow{3}{*}{0} & \multirow{3}{*}{137} \\
 & ZSTD        &   & 1.10 &   &  &  &  \\
 & FPZIP       &   & 4.06 &   &  &  &  \\
\cline{2-8}
 & ZFP         & $\epsilon=5$E-3 & 30.90 &  89.60 &  3821210 & 863116 & 438486 \\
 & SZ3         & $\epsilon=1$E-6 & 30.36 & 132.97 & 422 & 72 & 137 \\
 & cpSZ(SoS)   & $\epsilon=5$E-3 & 11.78 & 58.82 & 0 & 218539 & 2205 \\
\cline{2-8}
 & \multirow{3}{*}{Ours}  & Default $\epsilon=5$E-3 & 19.48 & 58.81 & \multirow{3}{*}{0} & \multirow{3}{*}{0} & \multirow{3}{*}{137} \\
 &             & SL $\epsilon=5$E-3      & 30.26 & 58.03 &  &  &  \\
 &             & MoP $\epsilon=5$E-3     & 31.71 & 58.05&  &  &  \\
\hline
\end{tabular}
}
\end{table}

\subsection{Comparison with State-of-the-Art Methods}
We present the quantitative results for all datasets in Table~\ref{tab:all_datasets}. The results include an ablation study of the predictor component, in which we evaluate three configurations of our method. The default configuration employs the standard 3D Lorenzo predictor (3DL); the SL configuration replaces it with the proposed semi-Lagrangian predictor; and the MoP configuration adaptively selects between the two predictors.  
All reported experimental runtimes correspond to serial execution without any parallelization. To ensure a fair comparison, we adjusted the parameters of SZ3 and cpSZ so that their compression ratios are as close as possible to that of our method in the default (3DL) configuration. According to the results, since SZ3 does not provide any guarantee of preserving critical points, it consequently fails to ensure the correctness of critical-point trajectories. Although cpSZ can preserve the critical points at each individual time step, it cannot guarantee the consistency of critical points within a time slab, which leads to inconsistent trajectory counts and the emergence of false cases within the slab. Among all the lossless compression frameworks, FPZIP achieves the best performance, with a maximum compression ratio of $4.06$. The results of GZIP and ZSTD are comparable, with compression ratios ranging between $1.1 $ and $ 1.3$.

On the other hand, all three of our proposed methods effectively preserve critical-point trajectories. Among them, MoP consistently outperforms 3DL across all four datasets, with particularly notable improvements on the CFVKV and FS datasets, achieving $72.29\times$ and $124.48\times$ compression ratios, respectively. Furthermore, when using SL alone as the predictor, it surpasses 3DL on three of the datasets with a maximum improvement of 32\%, and performs comparably to 3DL on the SCF dataset. Overall, among all compressors capable of preserving critical-point trajectories, our method achieve up to $56.07\times$ in compression ratio compared to the state-of-the-art compressors across all datasets. 
We observe that in the SCF dataset, a large fraction of vertices have $d_{\infty}$ significantly exceeding $d_{\max}$. Specifically, 85.4\% and 69.0\% of the vertices trigger the adaptive substepping mechanism in SL and MoP methods, respectively, with an average of 4.1 and 3.7 substeps. No significant increase in runtime has been observed,  indicating that the adaptive substepping mechanism does not introduce significant overhead in practice.

Figure~\ref{fig:time_comparison} shows the compression and decompression time across all the compressors. ZSTD is the fastest among all lossless compressors, while ZFP achieves the highest speed among lossy compressors. Compared with cpSZ, our methods are more time-consuming because they aim to preserve trajectories, which requires examining up to 36 surrounding triangles for each point, including pre-computation of critical points and derivation of error bounds. Among our three proposed methods, 3DL and SL exhibit similar compression times, whereas MoP requires more time due to the additional block partitioning, sampling, and micro-encoding processes used for score calculation. For decompression time, 3DL and SL perform almost identically, while MoP incurs a slightly higher cost because it must switch predictors across different data blocks.

Although our algorithm is slower than some existing lossy compressors, it uniquely preserves the critical-point trajectories throughout the compression–decompression process, ensuring topological consistency in the reconstructed vector fields. Meanwhile, the framework is highly amenable to parallelization, which can substantially improve its performance. Moreover, our SL predictor achieves an even higher degree of parallelism than the traditional Lorenzo predictor, since it is embarrassingly parallel across all data within frame~$t$ without depending on any current-frame neighbors. Consequently, existing studies~\cite{cusz,cuszp,tian2021cuszplus,liu_tian_wu2024cuszi,songceresz,qiu2025} have demonstrated that prediction-based compressors can further accelerate throughput significantly on GPUs and other hardware accelerators. In addition, the decompression process in our framework is relatively fast, and since data are typically compressed once but decompressed multiple times, the additional compression overhead has little impact on practical usability.

\subsection{Memory Consumption}

\begin{wrapfigure}{r}{1.7in}
    \includegraphics[width=1\linewidth]{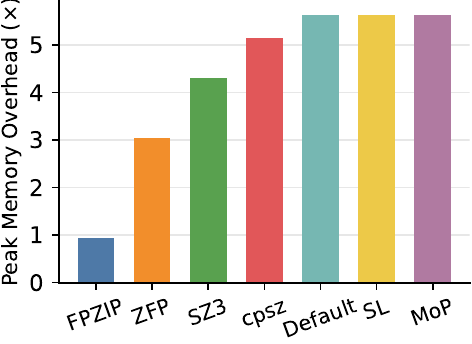}
    \vspace{-2em}
    \caption{Normalized peak memory usage of seven methods on average. Peak memory is normalized by the raw on-disk data size.}
    \label{fig:memory} 
    \vspace{-1em}
\end{wrapfigure}
Figure~\ref{fig:memory} reports the peak memory consumption during compression for our method and the compared compressors. GZIP and ZSTD are omitted since they are lossless, streaming-based compressors with low compression ratios and memory usage independent of data size. In contrast, predictor-based lossy compressors require additional working memory for data-aware prediction and auxiliary data structures. We therefore omit GZIP and ZSTD to focus on methods with comparable algorithmic characteristics. 

We further observe that FPZIP, ZFP, and SZ3 require approximately $1\times$, $3\times$, and $4\times$ the data size in memory, respectively, while both cpSZ-SoS and our method require slightly more than $5\times$. The additional memory usage of our method mainly stems from two factors: the use of int64 representations to support the SoS mechanism, which adds one extra data-size overhead, and an auxiliary hash map for maintaining critical-point information. We note that SoS is not inherently tied to this implementation choice and could alternatively be realized using arbitrary-precision libraries (e.g., GMP) with reduced storage overhead at the cost of higher computational expense. Our current design favors computational efficiency and simplicity in implementation.

\subsection{Qualitative Visualizations}
We visualize the trajectories of critical points in the original data and the decompressed data produced by our method (MoP), cpSZ, and SZ3 in Figure~\ref{fig:FS}. The error bound settings are kept consistent with those in Table~\ref{tab:all_datasets}.

Critical-point trajectories are extracted using the algorithm provided in FTK~\cite{ftk}. 
We overlay the line integral convolution (LIC) texture and the corresponding velocity field across time to display the critical-point trajectories. As shown in the visualizations, our method completely preserves the original critical-point trajectories. 
Both cpSZ and SZ3 distort certain trajectories in various ways, including but not limited to: (1) discontinuous trajectories accompanied by small oscillations, exhibiting zigzag patterns in the trajectories, which indicate slight positional deviations of critical points in the decompressed data; and (2) short-lived artifacts, manifested as ``small bubbles'' that suddenly appear, split and merge within a short time slab, where these critical points do not exist in the original data.

\section{Conclusion and Future Works}
\label{sec:conclusion}

In this paper, we address the problem of preserving critical-point trajectories during the lossy compression of 2D time-varying vector field data. We introduce a semi-Lagrangian–based predictor that achieves up to $4.31\times$ higher compression ratios than the Lorenzo predictor on advection-dominated datasets. Furthermore, we design a Mixture-of-Predictors module that adaptively selects the most suitable predictor for each local data block based on its spatial characteristics. This adaptive strategy effectively reduces residual entropy while maintaining high compression performance, achieving up to $56.07\times$ higher compression ratios than state-of-the-art lossless compressors in our experiments. 
In future work, we plan to further improve the efficiency of the framework, and to develop a GPU-accelerated implementation to achieve higher performance.

\section*{Acknowledgment}
This work was partially supported by grants from NSF OAC-2311756, OAC-2442627, OAC-2504255, OAC-2313122, OAC-2313123, OAC-2313124, OIA-2327266, and OAC-2311878.  
We thank the University of Kentucky Center for Computational Sciences and Information Technology Services Research Computing for its support and use of the Lipscomb Compute Cluster, Morgan Compute Cluster, and associated research computing resources.

\section*{AI-Generated Content Acknowledgment}
Portions of the text in this manuscript were refined using language-assistance tools such as Grammarly and ChatGPT (GPT-5). These tools were employed solely to improve grammar, clarity, and readability. All conceptualization, data analysis, and scientific conclusions are entirely the work of the authors.

\bibliographystyle{IEEEtran}
\bibliography{reference}

@article{liang2022toward,
	author = {Liang, Xin and Di, Sheng and Cappello, Franck and Raj, Mukund and Liu, Chunhui and Ono, Kenji and Chen, Zizhong and Peterka, Tom and Guo, Hanqi},
	journal = {IEEE Transactions on Visualization and Computer Graphics},
	number = {12},
	pages = {5434--5450},
	title = {Toward Feature-Preserving Vector Field Compression},
	volume = {29},
	year = {2023}}

@article{li2024msz,
	author = {Li, Yuxiao and Liang, Xin and Wang, Bei and Qiu, Yongfeng and Yan, Lin and Guo, Hanqi},
	journal = {IEEE Transactions on Visualization and Computer Graphics},
	number = {1},
	pages = {130--140},
	title = {{MSz}: An Efficient Parallel Algorithm for Correcting Morse-Smale Segmentations in Error-Bounded Lossy Compressors},
	volume = {31},
	year = {2024}}

@article{toposz,
	author = {Yan, Lin and Liang, Xin and Guo, Hanqi and Wang, Bei},
	doi = {10.1109/TVCG.2023.3326920},
	journal = {IEEE Transactions on Visualization and Computer Graphics},
	number = {1},
	pages = {1302---1312},
	title = {TopoSZ: Preserving Topology in Error-Bounded Lossy Compression},
	volume = {30},
	year = {2023}}

@manual{gzip,
	author = {Jean-loup Gailly and Mark Adler},
	note = {Accessed: January 24, 2025},
	organization = {Free Software Foundation},
	title = {gzip (GNU zip)},
	url = {https://www.gnu.org/software/gzip/},
	version = {1.13},
	year = {2023}}

@misc{zstd,
	author = {Yann Collet},
	howpublished = {\url{http://facebook.github.io/zstd/}},
	title = {Zstandard - Real-time data compression algorithm}}

@INPROCEEDINGS{tian2021cuszplus,
  author={Tian, Jiannan and Di, Sheng and Yu, Xiaodong and Rivera, Cody and Zhao, Kai and Jin, Sian and Feng, Yunhe and Liang, Xin and Tao, Dingwen and Cappello, Franck},
  booktitle={2021 IEEE International Conference on Cluster Computing (CLUSTER)}, 
  title={Optimizing Error-Bounded Lossy Compression for Scientific Data on GPUs}, 
  year={2021},
  volume={},
  number={},
  pages={283-293},
  doi={10.1109/Cluster48925.2021.00047}}

@inproceedings{liu_tian_wu2024cuszi,
	author = {Liu, Jinyang and Tian, Jiannan and Wu, Shixun and Di, Sheng and Zhang, Boyuan and Underwood, Robert and Huang, Yafan and Huang, Jiajun and Zhao, Kai and Li, Guanpeng and Tao, Dingwen and Chen, Zizhong and Cappello, Franck},
	booktitle = {SC '24: Proceedings of the International Conference for High Performance Computing, Networking, Storage and Analysis},
	title = {{\scshape cuSZ}-{\itshape i}: High-Ratio scientific lossy compression on GPUs with optimized multi-level interpolation},
	year = {2024}}

@article{zfp,
	author = {Lindstrom, Peter},
	journal = {IEEE Transactions on Visualization and Computer Graphics },
	number = {12},
	pages = {2674--2683},
	title = {Fixed-rate compressed floating-point arrays},
	volume = {20},
	year = {2014}}

@inproceedings{sz17,
	author = {Tao, Dingwen and Di, Sheng and Chen, Zizhong and Cappello, Franck},
	booktitle = {2017 IEEE International Parallel and Distributed Processing Symposium},
	organization = {IEEE},
	pages = {1129--1139},
	title = {Significantly improving lossy compression for scientific data sets based on multidimensional prediction and error-controlled quantization},
	year = {2017}}

@inproceedings{sz18,
	author = {Liang, Xin and Di, Sheng and Tao, Dingwen and Li, Sihuan and Li, Shaomeng and Guo, Hanqi and Chen, Zizhong and Cappello, Franck},
	booktitle = {2018 IEEE International Conference on Big Data},
	pages = {438--447},
	title = {Error-controlled lossy compression optimized for high compression ratios of scientific datasets},
	year = {2018}}

@inproceedings{zhao2021optimizing,
	author = {Zhao, Kai and Di, Sheng and Dmitriev, Maxim and Tonellot, Thierry-Laurent D. and Chen, Zizhong and Cappello, Franck},
	booktitle = {2021 IEEE 37th International Conference on Data Engineering (ICDE)},
	organization = {IEEE},
	pages = {1643--1654},
	title = {Optimizing Error-Bounded Lossy Compression for Scientific Data by Dynamic Spline Interpolation},
	year = {2021}}

@ARTICLE{liang2022sz3,
  author={Liang, Xin and Zhao, Kai and Di, Sheng and Li, Sihuan and Underwood, Robert and Gok, Ali M. and Tian, Jiannan and Deng, Junjing and Calhoun, Jon C. and Tao, Dingwen and Chen, Zizhong and Cappello, Franck},
  journal={IEEE Transactions on Big Data}, 
  title={SZ3: A Modular Framework for Composing Prediction-Based Error-Bounded Lossy Compressors}, 
  year={2023},
  volume={9},
  number={2},
  pages={485-498}}

@misc{mcc,
	howpublished = {\url{https://docs.ccs.uky.edu}},
	title = {{Morgan Compute Cluster}},
	year = {2023}}

@techreport{lindstrom2017error,
	author = {Lindstrom, Peter},
	institution = {Lawrence Livermore National Lab.(LLNL), Livermore, CA (United States)},
	title = {Error distributions of lossy floating-point compressors},
	year = {2017}}

@article{son2014data,
	author = {Son, Seung Woo and Chen, Zhengzhang and Hendrix, William and Agrawal, Ankit and Liao, Wei-keng and Choudhary, Alok},
	journal = {Supercomputing frontiers and innovations},
	number = {2},
	pages = {76--88},
	title = {Data compression for the exascale computing era-survey},
	volume = {1},
	year = {2014}}

@inproceedings{cuszp,
	author = {Huang, Yafan and Di, Sheng and Yu, Xiaodong and Li, Guanpeng and Cappello, Franck},
	booktitle = {Proceedings of the International Conference for High Performance Computing, Networking, Storage and Analysis},
	doi = {10.1145/3581784.3607048},
	title = {{cuSZp}: An Ultra-fast GPU Error-bounded Lossy Compression Framework with Optimized End-to-End Performance},
	year = {2023}}

@inproceedings{soler2018topologically,
	author = {M. Soler and M. Plainchault and B. Conche and J. Tierny},
	booktitle = {Proceedings of 2018 IEEE Pacific Visualization Symposium},
	doi = {10.1109/PacificVis.2018.00015},
	pages = {46--55},
	title = {Topologically Controlled Lossy Compression},
	year = {2018}}

@inproceedings{liang2018efficient,
	author = {Liang, Xin and Di, Sheng and Tao, Dingwen and Chen, Zizhong and Cappello, Franck},
	booktitle = {2018 IEEE International Conference on Cluster Computing (CLUSTER)},
	organization = {IEEE},
	pages = {179--189},
	title = {An efficient transformation scheme for lossy data compression with point-wise relative error bound},
	year = {2018}}

@inproceedings{sz16,
	//organization = {IEEE},
	address = {Chicago, IL, USA},
	author = {Di, Sheng and Cappello, Franck},
	booktitle = {2016 IEEE International Parallel and Distributed Processing Symposium},
	pages = {730--739},
	publisher = {IEEE},
	title = {Fast error-bounded lossy HPC data compression with {SZ}},
	year = {2016}}

@inproceedings{cusz,
	author = {Tian, Jiannan and Di, Sheng and Zhao, Kai and Rivera, Cody and Fulp, Megan Hickman and Underwood, Robert and Jin, Sian and Liang, Xin and Calhoun, Jon and Tao, Dingwen and others},
	booktitle = {Proceedings of the ACM International Conference on Parallel Architectures and Compilation Techniques},
	pages = {3--15},
	title = {Cusz: An efficient gpu-based error-bounded lossy compression framework for scientific data},
	year = {2020}}

@misc{sz-framework,
	author = {{UChicago Argonne LLC and WSU}},
	howpublished = {\url{https://szcompressor.org/}},
	title = {{SZ} Lossy Compressor: An open and transparent lossy compression framework}}

@article{lindstrom2014fixed,
	author = {Lindstrom, Peter},
	journal = {IEEE Transactions on Visualization and Computer Graphics},
	number = {12},
	pages = {2674--2683},
	title = {Fixed-Rate Compressed Floating-Point Arrays},
	volume = {20},
	year = {2014}}

@misc{lz4,
	author = {Collet, Yann},
	howpublished = {\url{https://github.com/lz4/lz4}},
	title = {{LZ4}}}

@article{ballester2019tthresh,
	author = {Ballester-Ripoll, Rafael and Lindstrom, Peter and Pajarola, Renato},
	journal = {IEEE transactions on visualization and computer graphics},
	number = {9},
	pages = {2891--2903},
	title = {{TTHRESH}: Tensor compression for multidimensional visual data},
	volume = {26},
	year = {2019}}

@inproceedings{ibarria2003out,
	author = {Ibarria, Lawrence and Lindstrom, Peter and Rossignac, Jarek and Szymczak, Andrzej},
	booktitle = {Computer Graphics Forum},
	number = {3},
	organization = {Wiley Online Library},
	pages = {343--348},
	title = {Out-of-core compression and decompression of large n-dimensional scalar fields},
	volume = {22},
	year = {2003}}

@misc{frontier,
	howpublished = {\url{https://www.olcf.ornl.gov/frontier}},
	title = {Frontier exscale supercomputer}}

@article{ainsworth2019multilevel,
	author = {Ainsworth, Mark and Tugluk, Ozan and Whitney, Ben and Klasky, Scott},
	journal = {SIAM Journal on Scientific Computing},
	number = {2},
	pages = {A1278--A1303},
	publisher = {SIAM},
	title = {Multilevel techniques for compression and reduction of scientific data---The multivariate case},
	volume = {41},
	year = {2019}}

@article{ainsworth2019qoi,
	author = {Ainsworth, Mark and Tugluk, Ozan and Whitney, Ben and Klasky, Scott},
	journal = {SIAM Journal on Scientific Computing},
	number = {4},
	pages = {A2146--A2171},
	publisher = {SIAM},
	title = {Multilevel techniques for compression and reduction of scientific data-quantitative control of accuracy in derived quantities},
	volume = {41},
	year = {2019}}

@article{ainsworth2018multilevel,
	author = {Ainsworth, Mark and Tugluk, Ozan and Whitney, Ben and Klasky, Scott},
	journal = {Computing and Visualization in Science},
	number = {5-6},
	pages = {65--76},
	publisher = {Springer},
	title = {Multilevel techniques for compression and reduction of scientific data---the univariate case},
	volume = {19},
	year = {2018}}

@inproceedings{liang2020toward,
  author={Liang, Xin and Guo, Hanqi and Di, Sheng and Cappello, Franck and Raj, Mukund and Liu, Chunhui and Ono, Kenji and Chen, Zizhong and Peterka, Tom},
  booktitle={2020 IEEE Pacific Visualization Symposium (PacificVis)}, 
  title={Toward Feature-Preserving 2D and 3D Vector Field Compression}, 
  year={2020},
  pages={81-90}}

@ARTICLE{liang2021mgard+,
  author={Liang, Xin and Whitney, Ben and Chen, Jieyang and Wan, Lipeng and Liu, Qing and Tao, Dingwen and Kress, James and Pugmire, David and Wolf, Matthew and Podhorszki, Norbert and Klasky, Scott},
  journal={IEEE Transactions on Computers}, 
  title={MGARD+: Optimizing Multilevel Methods for Error-Bounded Scientific Data Reduction}, 
  year={2022},
  volume={71},
  number={7},
  pages={1522-1536}}

@article{lorensen1987marching,
	author = {Lorensen, William E and Cline, Harvey E},
	journal = {ACM siggraph computer graphics},
	number = {4},
	pages = {163--169},
	publisher = {ACM New York, NY, USA},
	title = {Marching cubes: A high resolution 3D surface construction algorithm},
	volume = {21},
	year = {1987}}

@article{SCALE,
	author = {Guo-Yuan Lien and Takemasa Miyoshi and Seiya Nishizawa and Ryuji Yoshida and Hisashi Yashiro and Sachiho A. Adachi and Tsuyoshi Yamaura and Hirofumi Tomita},
	doi = {10.2151/sola.2017-001},
	journal = {SOLA},
	pages = {1-6},
	title = {The Near-Real-Time SCALE-LETKF System: A Case of the September 2015 Kanto-Tohoku Heavy Rainfall},
	volume = {13},
	year = {2017}}

@inproceedings{li2023lossy,
	author = {Li, Shaomeng and Lindstrom, Peter and Clyne, John},
	booktitle = {2023 IEEE International Parallel and Distributed Processing Symposium (IPDPS)},
	organization = {IEEE},
	pages = {1007--1017},
	title = {Lossy scientific data compression with SPERR},
	year = {2023}}

@inproceedings{xia2024preserving,
	author = {Xia, Mingze and Di, Sheng and Cappello, Franck and Jiao, Pu and Zhao, Kai and Liu, Jinyang and Wu, Xuan and Liang, Xin and Guo, Hanqi},
	booktitle = {2024 IEEE 40th International Conference on Data Engineering (ICDE)},
	organization = {IEEE},
	pages = {4979--4992},
	title = {Preserving Topological Feature with Sign-of-Determinant Predicates in Lossy Compression: A Case Study of Vector Field Critical Points},
	year = {2024}}

@article{sos,
	author = {Edelsbrunner, Herbert and M\"{u}cke, Ernst Peter},
	doi = {10.1145/77635.77639},
	journal = {ACM Transactions on Graphics},
	number = {1},
	pages = {66---104},
	title = {Simulation of simplicity: a technique to cope with degenerate cases in geometric algorithms},
	volume = {9},
	year = {1990}}

@book{tetrahedron_split,
	author = {De Loera, Jes{\'u}s and Rambau, J{\"o}rg and Santos, Francisco},
	publisher = {Springer Science \& Business Media},
	title = {Triangulations: structures for algorithms and applications},
	volume = {25},
	year = {2010}}

@article{Jung93,
	author = {Jung, C. and T\'el, T. and Ziemniak, E.},
	doi = {10.1063/1.165960},
	journal = {Chaos: An Interdisciplinary Journal of Nonlinear Science},
	number = {4},
	pages = {555--568},
	title = {Application of scattering chaos to particle transport in a hydrodynamical flow},
	volume = {3},
	year = {1993}}

@article{gerrisflowsolver,
	author = {S. Popinet},
	journal = {ClusterWorld},
	number = {6},
	title = {Free Computational Fluid Dynamics},
	url = {http://gfs.sf.net/},
	volume = {2},
	year = {2004}}

@article{Guenther17,
	author = {Tobias G{\"u}nther and Markus Gross and Holger Theisel},
	journal = {ACM Transactions on Graphics (Proc. SIGGRAPH)},
	location = {Los Angeles, United States},
	number = {4},
	pages = {141:1--141:11},
	title = {Generic Objective Vortices for Flow Visualization},
	volume = {36},
	year = {2017}}

@article{Jakob20,
	author = {Jakob Jakob and Markus Gross and Tobias G{\"u}nther},
	journal = {IEEE Transactions on Visualization and Computer Graphics},
	number = {2},
	pages = {1279--1289},
	title = {A Fluid Flow Data Set for Machine Learning and its Application to Neural Flow Map Interpolation},
	volume = {27},
	year = {2021}}

@inproceedings{fff,
	address = {Goslar, DEU},
	author = {Theisel, H. and Seidel, H.-P.},
	booktitle = {Proceedings of the Symposium on Data Visualisation 2003},
	isbn = {1581136986},
	location = {Grenoble, France},
	numpages = {8},
	pages = {141--148},
	publisher = {Eurographics Association},
	series = {VISSYM '03},
	title = {Feature flow fields},
	year = {2003}}

@misc{multitierMSZ,
	author = {Yuxiao Li and Mingze Xia and Xin Liang and Bei Wang and Hanqi Guo},
	howpublished = {arXiv preprint arXiv:2409.17346},
	title = {Multi-Tier Preservation of Discrete {Morse-Smale} Complexes in Error-Bounded Lossy Compression},
	year = {2025}}

@article{GeneralizedTopo,
	author = {Tierny, Julien and Pascucci, Valerio},
	doi = {10.1109/TVCG.2012.228},
	journal = {IEEE Transactions on Visualization and Computer Graphics},
	number = {12},
	pages = {2005--2013},
	title = {Generalized Topological Simplification of Scalar Fields on Surfaces},
	volume = {18},
	year = {2012}}

@inproceedings{xia2025tspsz,
	author = {Xia, Mingze and Wang, Bei and Li, Yuxiao and Jiao, Pu and Liang, Xin and Guo, Hanqi},
	booktitle = {2025 IEEE 41st International Conference on Data Engineering (ICDE)},
	doi = {10.1109/ICDE65448.2025.00275},
	pages = {3682-3695},
	title = {TspSZ: An Efficient Parallel Error-Bounded Lossy Compressor for Topological Skeleton Preservation},
	year = {2025}}

@article{ftk,
	author = {Guo, Hanqi and Lenz, David and Xu, Jiayi and Liang, Xin and He, Wenbin and Grindeanu, Iulian R. and Shen, Han-Wei and Peterka, Tom and Munson, Todd and Foster, Ian},
	doi = {10.1109/TVCG.2021.3073399},
	journal = {IEEE Transactions on Visualization and Computer Graphics},
	number = {8},
	pages = {3463--3480},
	title = {{FTK}: A Simplicial Spacetime Meshing Framework for Robust and Scalable Feature Tracking},
	volume = {27},
	year = {2021}}

@article{pu2022qoi,
	author = {Jiao, Pu and Di, Sheng and Guo, Hanqi and Zhao, Kai and Tian, Jiannan and Tao, Dingwen and Liang, Xin and Cappello, Franck},
	doi = {10.14778/3574245.3574255},
	journal = {Proceedings of the VLDB Endowment},
	number = {4},
	pages = {697--710},
	title = {Toward Quantity-of-Interest Preserving Lossy Compression for Scientific Data},
	volume = {16},
	year = {2022}}

@article{liu2024qoi,
	author = {Liu, Jinyang and Jiao, Pu and Zhao, Kai and Liang, Xin and Di, Sheng and Cappello, Franck},
	doi = {10.14778/3742728.3742739},
	journal = {Proceedings of the VLDB Endowment },
	number = {8},
	pages = {2440---2453},
	title = {{QPET}: A Versatile and Portable Quantity-of-Interest-Preservation Framework for Error-Bounded Lossy Compression},
	volume = {18},
	year = {2025}}

@inproceedings{li2023sperr,
	author = {Li, Shaomeng and Lindstrom, Peter and Clyne, John},
	booktitle = {2023 IEEE International Parallel and Distributed Processing Symposium (IPDPS)},
	organization = {IEEE},
	pages = {1007--1017},
	title = {Lossy scientific data compression with SPERR},
	year = {2023}}

@article{liu2024high,
	author = {Liu, Jinyang and Di, Sheng and Zhao, Kai and Liang, Xin and Jin, Sian and Jian, Zizhe and Huang, Jiajun and Wu, Shixun and Chen, Zizhong and Cappello, Franck},
	journal = {Proceedings of the ACM on Management of Data},
	number = {1},
	pages = {1--27},
	publisher = {ACM New York, NY, USA},
	title = {High-performance effective scientific error-bounded lossy compression with auto-tuned multi-component interpolation},
	volume = {2},
	year = {2024}}

@article{yeung2025small,
	author = {Yeung, PK and Ravikumar, Kiran and Uma-Vaideswaran, Rohini and Dotson, Daniel L and Sreenivasan, Katepalli R and Pope, Stephen B and Meneveau, Charles and Nichols, Stephen},
	journal = {Journal of Fluid Mechanics},
	pages = {R2},
	publisher = {Cambridge University Press},
	title = {Small-scale properties from exascale computations of turbulence on a periodic cube},
	volume = {1019},
	year = {2025}}

@article{hersbach2020era5,
	author = {Hersbach, Hans and Bell, Bill and Berrisford, Paul and Hirahara, Shoji and Hor{\'a}nyi, Andr{\'a}s and Mu{\~n}oz-Sabater, Joaqu{\'\i}n and Nicolas, Julien and Peubey, Carole and Radu, Raluca and Schepers, Dinand and others},
	journal = {Quarterly journal of the royal meteorological society},
	number = {730},
	pages = {1999--2049},
	publisher = {Wiley Online Library},
	title = {The {ERA5} global reanalysis},
	volume = {146},
	year = {2020}}

@article{liu2025solarzip,
	author = {Liu, Zedong and Tan, Song and Warmuth, Alexander and Schuller, Fr{\'e}d{\'e}ric and Hong, Yun and Huang, Wenjing and Gu, Yida and Zhu, Bojing and Tan, Guangming and Tao, Dingwen},
	journal = {Astronomy \& Astrophysics},
	pages = {A160},
	publisher = {EDP Sciences},
	title = {SolarZip: An efficient and adaptive compression framework for Solar EUV imaging data-Application to Solar Orbiter/EUI data},
	volume = {702},
	year = {2025}}

@article{fpzip,
	author = {Lindstrom, Peter and Isenburg, Martin},
	doi = {10.1109/TVCG.2006.143},
	journal = {IEEE Transactions on Visualization and Computer Graphics},
	number = {5},
	pages = {1245--1250},
	title = {Fast and Efficient Compression of Floating-Point Data},
	volume = {12},
	year = {2006}}

@inproceedings{liu2023faz,
	author = {Liu, Jinyang and Di, Sheng and Zhao, Kai and Liang, Xin and Chen, Zizhong and Cappello, Franck},
	booktitle = {Proceedings of the 37th ACM International Conference on Supercomputing},
	doi = {10.1145/3577193.3593721},
	pages = {1---13},
	title = {{FAZ}: A flexible auto-tuned modular error-bounded compression framework for scientific data},
	year = {2023}}

@article{TFZ,
	author = {Gorski, Nathaniel and Guo, Hanqi and Liang, Xin and Wang, Bei},
	journal = {IEEE Transactions on Visualization and Computer Graphics},
	number = {1},
	pages = {527--537},
	title = {{TFZ}: Topology-Preserving Compression of 2D Symmetric and Asymmetric Second-Order Tensor Fields},
	volume = {32},
	year = {2026}}

@article{gunther2017generic,
	author = {G{\"u}nther, Tobias and Gross, Markus and Theisel, Holger},
	journal = {ACM Transactions on Graphics (TOG)},
	number = {4},
	pages = {1--11},
	publisher = {ACM New York, NY, USA},
	title = {Generic objective vortices for flow visualization},
	volume = {36},
	year = {2017}}

@incollection{friederici2015finite,
	author = {Friederici, Anke and R{\"o}ssl, Christian and Theisel, Holger},
	booktitle = {Topological Methods in Data Analysis and Visualization},
	pages = {253--266},
	publisher = {Springer},
	title = {Finite time steady 2D vector field topology},
	year = {2015}}

@inproceedings{fritschi2019visualizing,
  author={Fritschi, Lea and Rojo, Irene Baeza and Günther, Tobias},
  booktitle={2019 IEEE Scientific Visualization Conference (SciVis)}, 
  title={Visualizing the Temporal Evolution of the Universe from Cosmology Simulations}, 
  year={2019},
  volume={},
  number={},
  pages={1-11}}

@article{rimensberger2019visualization,
	author = {Rimensberger, No{\"e}l and Gross, Markus and G{\"u}nther, Tobias},
	journal = {IEEE Computer Graphics and Applications},
	number = {1},
	pages = {12--25},
	publisher = {IEEE},
	title = {Visualization of clouds and atmospheric air flows},
	volume = {39},
	year = {2019}}

@misc{Basilisk,
	author = {Popinet, St{\'e}phane},
	howpublished = {\url{http://basilisk.fr}},
	note = {Accessed: 2024},
	title = {Basilisk: a free software for fluid dynamics},
	year = {2013}}

@article{TRICOCHE2002249,
	author = {X Tricoche and T Wischgoll and G Scheuermann and H Hagen},
	journal = {Computers \& Graphics},
	number = {2},
	pages = {249--257},
	title = {Topology tracking for the visualization of time-dependent two-dimensional flows},
	volume = {26},
	year = {2002}}

@article{1432684,
	author = {Theisel, H. and Weinkauf, T. and Hege, H.-C. and Seidel, H.-P.},
	doi = {10.1109/TVCG.2005.68},
	journal = {IEEE Transactions on Visualization and Computer Graphics},
	number = {4},
	pages = {383--394},
	title = {Topological methods for 2D time-dependent vector fields based on stream lines and path lines},
	volume = {11},
	year = {2005}}

@inproceedings{nielson1991asymptotic,
	author = {Nielson, G.M. and Hamann, B.},
	booktitle = {Proceeding Visualization '91},
	doi = {10.1109/VISUAL.1991.175782},
	pages = {83-91},
	title = {The asymptotic decider: resolving the ambiguity in marching cubes},
	year = {1991}}

@book{de2008computational,
	author = {De Berg, Mark and Cheong, Otfried and Van Kreveld, Marc and Overmars, Mark},
	publisher = {Springer},
	title = {Computational geometry: algorithms and applications},
	year = {2008}}

@article{neu2013imilast,
	author = {Neu, Urs and Akperov, Mirseid G and Bellenbaum, Nina and Benestad, Rasmus and Blender, Richard and Caballero, Rodrigo and Cocozza, Angela and Dacre, Helen F and Feng, Yang and Fraedrich, Klaus and others},
	journal = {Bulletin of the American Meteorological Society},
	number = {4},
	pages = {529--547},
	publisher = {American Meteorological Society},
	title = {IMILAST: A community effort to intercompare extratropical cyclone detection and tracking algorithms},
	volume = {94},
	year = {2013}}

@article{woollings2018blocking,
	author = {Woollings, Tim and Barriopedro, David and Methven, John and Son, Seok-Woo and Martius, Olivia and Harvey, Ben and Sillmann, Jana and Lupo, Anthony R and Seneviratne, Sonia},
	journal = {Current climate change reports},
	number = {3},
	pages = {287--300},
	publisher = {Springer},
	title = {Blocking and its response to climate change},
	volume = {4},
	year = {2018}}

@article{bai2020time,
	author = {Bai, Zhihui and Tao, Yubo and Lin, Hai},
	journal = {Journal of Visualization},
	number = {5},
	pages = {745--761},
	publisher = {Springer},
	title = {Time-varying volume visualization: a survey},
	volume = {23},
	year = {2020}}

@article{haller2005objective,
	author = {Haller, George},
	journal = {Journal of fluid mechanics},
	pages = {1--26},
	publisher = {Cambridge University Press},
	title = {An objective definition of a vortex},
	volume = {525},
	year = {2005}}

@article{hewson2010objective,
	author = {Hewson, Tim D and Titley, Helen A},
	journal = {Meteorological Applications},
	number = {3},
	pages = {355--381},
	publisher = {Wiley Online Library},
	title = {Objective identification, typing and tracking of the complete life-cycles of cyclonic features at high spatial resolution},
	volume = {17},
	year = {2010}}

@article{nencioli2010vector,
	author = {Nencioli, Francesco and Dong, Changming and Dickey, Tommy and Washburn, Libe and McWilliams, James C},
	journal = {Journal of atmospheric and oceanic technology},
	number = {3},
	pages = {564--579},
	title = {A vector geometry--based eddy detection algorithm and its application to a high-resolution numerical model product and high-frequency radar surface velocities in the Southern California Bight},
	volume = {27},
	year = {2010}}

@article{BERSON200961,
	author = {Argantha{\"e}l Berson and Marc Michard and Philippe Blanc-Benon},
	journal = {Comptes Rendus M{\'e}canique},
	number = {2},
	pages = {61-67},
	title = {Vortex identification and tracking in unsteady flows},
	volume = {337},
	year = {2009}}

@article{CoresofSwirling,
	author = {Weinkauf, Tino and Sahner, Jan and Theisel, Holger and Hege, Hans-Christian},
	doi = {10.1109/TVCG.2007.70545},
	journal = {IEEE Transactions on Visualization and Computer Graphics},
	number = {6},
	pages = {1759--1766},
	title = {Cores of Swirling Particle Motion in Unsteady Flows},
	volume = {13},
	year = {2007}}

@inproceedings{liu2025hpdc,
	author = {Liu, Youyuan and Jia, Wenqi and Yang, Taolue and Jiang, Bo and Yin, Miao and Jin, Sian},
	booktitle = {Proceedings of the 34th International Symposium on High-Performance Parallel and Distributed Computing},
	doi = {10.1145/3731545.3731592},
	title = {Advancing Scientific Data Compression via Cross-Field Prediction},
	year = {2025}}

@inproceedings{hu2025icde,
	author = {Hu, Hao and Zheng, Qiyang and Zou, Xiangyu and Qin, Lisha and Zhang, Chengwei and Zhang, Wanchuan and Jiang, Zhaoheng and Tao, Dingwen and Wang, Hongpeng and Xia, Wen},
	booktitle = {2025 IEEE 41st International Conference on Data Engineering (ICDE)},
	doi = {10.1109/ICDE65448.2025.00037},
	pages = {405-418},
	title = {A Cost-Effective and Decompression-Transparent Compressor for OLTP-Oriented Databases},
	year = {2025}}

@inproceedings{wuxuan2025ipdps,
	author = {Wu, Xuan and Di, Sheng and Ren, Congrong and Jiao, Pu and Xia, Mingze and Wang, Cheng and Guo, Hanqi and Liang, Xin and Cappello, Franck},
	booktitle = {2025 IEEE International Parallel and Distributed Processing Symposium (IPDPS)},
	doi = {10.1109/IPDPS64566.2025.00040},
	pages = {370-382},
	title = {Enabling Efficient Error-Controlled Lossy Compression for Unstructured Scientific Data},
	year = {2025}}

@inproceedings{wuxuan2025sc,
	author = {Wu, Xuan and Gong, Qian and Chen, Jieyang and Liu, Qing and Podhorszki, Norbert and Liang, Xin and Klasky, Scott},
	booktitle = {SC24: International Conference for High Performance Computing, Networking, Storage and Analysis},
	doi = {10.1109/SC41406.2024.00092},
	pages = {1-16},
	title = {Error-controlled Progressive Retrieval of Scientific Data under Derivable Quantities of Interest},
	year = {2024}}

@inproceedings{liu2025bittuner,
	author = {Liu, Qiyu and Luo, Yuxin and Cui, Mengke and Han, Siyuan and Peng, Jingshu and Li, Jin and Chen, Lei},
	booktitle = {2025 IEEE 41st International Conference on Data Engineering (ICDE)},
	doi = {10.1109/ICDE65448.2025.00349},
	pages = {4548-4551},
	title = {BitTuner: A Toolbox for Automatically Configuring Learned Data Compressors},
	year = {2025}}

@inproceedings{songceresz,
	author = {Song, Shihui and Huang, Yafan and Jiang, Peng and Yu, Xiaodong and Zheng, Weijian and Di, Sheng and Cao, Qinglei and Feng, Yunhe and Xie, Zhen and Cappello, Franck},
	booktitle = {Proceedings of the 33rd International Symposium on High-Performance Parallel and Distributed Computing},
	doi = {10.1145/3625549.3658691},
	pages = {309--321},
	title = {CereSZ: Enabling and Scaling Error-bounded Lossy Compression on Cerebras CS-2},
	year = {2024}}

@inproceedings{lixi2024icde,
	author = {Zhou, Lixi and Candan, K. Sel{\c c}uk and Zou, Jia},
	booktitle = {2024 IEEE 40th International Conference on Data Engineering (ICDE)},
	doi = {10.1109/ICDE60146.2024.00008},
	pages = {1-14},
	title = {DeepMapping: Learned Data Mapping for Lossless Compression and Efficient Lookup},
	year = {2024}}

@article{tan2024ts,
	author = {Tan, Haoliang and Xia, Wen and Zou, Xiangyu and Deng, Cai and Liao, Qing and Gu, Zhaoquan},
	doi = {10.1145/3664817},
	journal = {ACM Transactions on Storage},
	number = {4},
	title = {The Design of Fast {Delta} Encoding for Delta Compression Based Storage Systems},
	volume = {20},
	year = {2024}}

@inproceedings{tan2024dcc,
	author = {Tan, Haoliang and Zou, Xiangyu and Wan, Binzhaoshuo and Gu, Zhaoquan and Xia, Wen},
	booktitle = {2024 Data Compression Conference (DCC)},
	doi = {10.1109/DCC58796.2024.00044},
	pages = {362-371},
	title = {SuperDelta: Multiple Referenced Base Chunks Scheme for Fine-grained Deduplication Backup Storage System},
	year = {2024}}

@article{guoxi2025gale,
	author = {Liu, Guoxi and Randall, Thomas and Ge, Rong and Iuricich, Federico},
	doi = {10.1109/TVCG.2025.3634637},
	journal = {IEEE Transactions on Visualization and Computer Graphics},
	number = {1},
	pages = {1010--1020},
	title = {{GALE}: Leveraging Heterogeneous Systems for Efficient Unstructured Mesh Data Analysis},
	volume = {32},
	year = {2026}}

@article{guoxi2023atask,
	author = {Liu, Guoxi and Iuricich, Federico},
	doi = {10.1109/TVCG.2023.3327182},
	journal = {IEEE Transactions on Visualization and Computer Graphics},
	number = {1},
	pages = {1271--1281},
	title = {A Task-Parallel Approach for Localized Topological Data Structures},
	volume = {30},
	year = {2024}}

@inproceedings{qiu2025,
	author = {Qiu, Yongfeng and Li, Yuxiao and Liang, Xin and Huang, Yafan and Li, Guanpeng and Di, Sheng and Cappello, Franck and Guo, Hanqi},
	booktitle = {2025 IEEE International Parallel and Distributed Processing Symposium Workshops (IPDPSW)},
	doi = {10.1109/IPDPSW66978.2025.00213},
	pages = {1283-1285},
	title = {Lossy Parallel Visualization of Large-Scale Volume Data with Error-Bounded Image Compositing},
	year = {2025}}

@article{nathan2025pvis,
	author = {Gorski, Nathaniel and Liang, Xin and Guo, Hanqi and Yan, Lin and Wang, Bei},
	doi = {10.1109/TVCG.2025.3567054},
	journal = {IEEE Transactions on Visualization and Computer Graphics},
	number = {6},
	pages = {3693--3705},
	title = {A General Framework for Augmenting Lossy Compressors With Topological Guarantees},
	volume = {31},
	year = {2025}}

\end{document}